\documentclass[11pt,oneside,fleqn]{article}

\usepackage[ansinew]{inputenc}
\usepackage{hyperref}
\usepackage{amsmath,amssymb,amsthm}
\usepackage{graphicx}
\usepackage{color}

\allowdisplaybreaks
\hypersetup{colorlinks, linkcolor=blue, citecolor=blue, urlcolor=blue}
\newcommand{\pdl}[2]{\frac{\partial #1}{\partial #2}}

\newcommand{\ddd}{\mathrm{d}}
\newcommand{\p}{\partial}

\newcommand{\const}{\mathop{\rm const}\nolimits}
\newcommand{\todo}[1][\null]{\ensuremath{\clubsuit}}

\newcommand{\noprint}[1]{}

\newtheorem{theorem}{Theorem}

{\theoremstyle{definition}

\newtheorem{remark}{Remark}
\newtheorem*{remark*}{Remark}
\newtheorem*{example*}{Example}
}

\newcommand{\checked}[1][\null]{\ensuremath{\boldsymbol{\surd}}}

\newcommand{\DD}{\mathrm{D}}

\newcommand{\nn}{\mathbf{\nabla}}

\newcommand{\ve}{\varepsilon}

\newcommand{\Ad}{\mathrm{Ad}}

\newcommand{\FF}{\mathcal{F}}
\newcommand{\ZZ}{\mathcal{Z}}

\newcommand{\XX}{\mathcal{X}}

\begin{document}

\par\noindent {\LARGE\bf
Lie symmetry analysis and exact solutions \\ of the quasi-geostrophic two-layer problem
\par}
{\vspace{4mm}\par\noindent {\bf Alexander Bihlo~$^\dag$ and Roman O. Popovych~$^\dag\, ^\ddag$
} \par\vspace{2mm}\par}

{\vspace{2mm}\par\noindent {\it
$^{\dag}$~Faculty of Mathematics, University of Vienna, Nordbergstra{\ss}e 15, A-1090 Vienna, Austria\\
}}
{\noindent \vspace{2mm}{\it
$\phantom{^\dag}$~\textup{E-mail}: alexander.bihlo@univie.ac.at
}\par}

{\vspace{2mm}\par\noindent {\it
$^\ddag$~Institute of Mathematics of NAS of Ukraine, 3 Tereshchenkivska Str., 01601 Kyiv, Ukraine\\
}}
{\noindent \vspace{2mm}{\it
$\phantom{^\dag}$~\textup{E-mail}: rop@imath.kiev.ua
}\par}

\vspace{2mm}\par\noindent\hspace*{8mm}\parbox{140mm}{\small
The quasi-geostrophic two-layer model is of superior interest in dynamic meteorology since it is one of the easiest ways to study baroclinic processes in geophysical fluid dynamics. The complete set of point symmetries of the two-layer equations is determined. An optimal set of one- and two-dimensional inequivalent subalgebras of the maximal Lie invariance algebra is constructed. On the basis of these subalgebras we exhaustively carry out group-invariant reduction and compute various classes of exact solutions. Where possible, reference to the physical meaning of the exact solutions is given. In particular, the well-known baroclinic Rossby wave solutions in the two-layer model are rediscovered. 
}\par\vspace{2mm}

\section{Introduction}

There is a long history in dynamic meteorology to use simplified models of the atmosphere rather than the complete set of hydro-thermodynamical equations to study only selected phenomenon instead of accounting for the whole variety of weather and climate at once. The greatest simplification which is still capable to qualitatively (and, under some conditions, also quantitatively) describe the behavior of large-scale geophysical fluid dynamics is the barotropic vorticity equation. While this equation indeed allows one to explain the propagation of the mid-latitude Rossby waves, it cannot be used to elucidate the occurrence of developing weather regimes. The reason for this substantial lack is that the barotropic vorticity equation is a single equation valid only in one atmospheric (pressure) layer. However, development in the atmosphere is usually associated with the vertical structure of, e.g.,~the entropy field and hence a single-layer consideration is at once limited.

Of course, the atmosphere is continuously stratified and hence it is in fact three-dimensional. However, the main process of \emph{baroclinic instability}, which is the dominant mechanism responsible for the formation of mid-latitude weather systems can already be qualitatively understood by considering only two coupled atmospheric layers. Moreover, the studies of layer models create a basis for the more complicated investigation of the three-dimensional governing equations. For this reason, many results for the two-layer model are available in dynamic meteorology \cite{holt04Ay,pedl87Ay}, making this model particularly interesting for a systematic mathematical investigation.

The aim of this paper is to carry out Lie symmetry analysis of the quasi-geostrophic two-layer model. There already exist a number of papers dealing with symmetries and exact solutions of the simpler barotropic vorticity equations \cite{bihl09Ay,bihl09Cy,ibra95Ay,katk65Ay,katk66Ay} as well as the more complicated Euler or Navier--Stokes equations (see, e.g., \cite{andr98Ay,fush94Ay,mele04Ay,mele05Ay,popo00Ay} and references therein). At the same time, the intermediate two-layer model has not been considered in this light so far.

The organization of the paper is the following:
The model equations are presented in Section~\ref{sec:modelQG}, together with some of their known properties.
Section~\ref{sec:symmetriesQG} contains the results on Lie symmetries of these equations.
The theorem on the complete point symmetry group of the two-layer model is proved in Section~\ref{sec:CompletePointSymmetryGroupQG}.
A list of one- and two-dimensional inequivalent subalgebras of the corresponding maximal Lie invariance algebra is constructed in Sections~\ref{sec:InequivalentSubalgebras}.
In Sections~\ref{sec:reduction1QG} and~\ref{sec:reduction2QG}, systems obtained under reduction upon using the constructed one- and two-dimensional subalgebras are derived and investigated. 
The Rossby wave solutions in the quasi-geostrophic two-layer model are recovered as Lie invariant solutions.
Related boundary-value problems are discussed in Section~\ref{sec:boundaryQG}.
Section~\ref{sec:conclusionQG} summarizes the most important results of this paper.
Finally, in the Appendix~\ref{sec:extendedsetofsolutionsQG} we give a new symmetry interpretation of a method for finding exact solutions of linear systems of PDEs.

\section{The two-layer model}\label{sec:modelQG}

The first impulse to the investigation of baroclinic instability in the two-layer model was given in~\cite{phil54Ay}. Since two layers are considered, the model is capable of studying the interaction of the barotropic mode and the first baroclinic mode, which is sufficient for describing the basic mechanism of baroclinic instability. The model consists of two copies of the barotropic potential vorticity, evaluated at two different atmospheric levels~\cite{pedl87Ay}:
\begin{subequations}\label{twolayer}
\begin{align}
\begin{split}
    &Q^1_t+ \{\psi^1,Q^1\} = 0, \\
    &Q^2_t+ \{\psi^2,Q^2\} = 0,
\end{split}
\end{align}
where $\psi^1$ and $\psi^2$ are the stream functions in the upper and lower layer,
\begin{align}
\begin{split}
    &Q^1 = \nabla^2\psi^1 + \beta y - F(\psi^1-\psi^2), \\
    &Q^2 = \nabla^2\psi^2 + \beta y + F(\psi^1-\psi^2),
\end{split}
\end{align}
\end{subequations}
are the corresponding potential vorticities, 
with the constants $\beta$ and $F$ being the Rossby parameter and internal rotational Froude number, respectively, and
\[\{\psi^\bullet,Q^\bullet\}= \psi^\bullet_xQ^\bullet_y-\psi^\bullet_yQ^\bullet_x\]
is the usual Poisson bracket of the functions $\psi^\bullet$ and $Q^\bullet$ with respect to the space variables $x$ and~$y$.
For simplicity, we set the Froude numbers of the two layers to be equal, i.e. $F_1=F_2=F$, thereby assuming both layers to be of equal depth. Moreover, flat topography is assumed. For the configuration to be stably stratified, the lower layer must be denser than the upper layer. 

Here and in what follows subscripts of functions denote differentiation with respect to the corresponding variables 
and prime denotes differentiation with respect to~$t$.

Due to equivalence transformations of scaling and alternating signs in the class of equations of form~\eqref{twolayer}, it would be possible to set $F\in\{-1,0,1\}$ and $\beta\in\{0,1\}$. Since $F$ and $\beta$ are positive meteorological quantities, it would imply that $F=\beta=1$ but we will not use this scaling below.

\section{Lie symmetries}\label{sec:symmetriesQG}

In the case $F\beta\ne0$, which is the single subject of the present paper, system~\eqref{twolayer} admits the maximal Lie invariance algebra~$\mathfrak g$ generated by the following basis elements:
\begin{equation}\label{eq:symmetries}
    \p_t, \quad \p_y,\quad
    \XX(f)=f(t)\p_x - f'(t)y(\p_{\psi^1}+\p_{\psi^2}),\quad
    \FF=\p_{\psi^1}-\p_{\psi^2},\quad
    \ZZ(g)=g(t)(\p_{\psi^1}+\p_{\psi^2}),
\end{equation}
where $f$ and $g$ run through the set of real-valued functions of $t$. The computation of the algebra~$\mathfrak g$ was checked using the symbolic calculation packages  \texttt{DESOLV}~\cite{carm00Ay} and \texttt{muLie}~\cite{head93Ay}. For $\mathfrak g$ to really be a Lie algebra, additional restrictions on the smoothness of the parameter functions $f$ and $g$ should be imposed~\cite{fush94Ay}.

Physically, these generators are exponentiated to give time and north--south translation symmetry, generalized Galilean transformations with respect to $x$, as well as shifts and gauging of the stream functions.

The structure of the Lie symmetry algebra suggests the introduction of the new dependent variables $\psi^+=\psi^1+\psi^2$ and $\psi^-=\psi^1-\psi^2$. This change transforms the set of generators to
\[
    \p_t, \quad \p_y,\quad \XX(f)=f(t)\p_x - 2f'(t)y\p_{\psi^+},\quad \FF=2\p_{\psi^-},\quad \ZZ(g)=2g(t)\p_{\psi^+}.
\]
From the meteorological point of view, the new variables have a sound physical meaning. Since the two-dimensional wind fields at both the levels can be represented as the derivatives of the respective stream functions, it follows that the component $\psi^+$ gives rise to the mean of these fields. This part is usually referred to as the barotropic part of the flow. In turn, derivatives of $\psi^-$ give rise to the difference in the wind field between the two-layers. In meteorology, this difference is called the \emph{thermal wind}, which is a measure of the baroclinity of the fluid. Therefore, these variables are commonly used in the investigation of baroclinic instability, e.g.,~in the study of the linearized two-layer model~\cite{holt04Ay}. However, from the group-theoretical point of view, their usage is already suggested by the special form of Lie symmetry operators~\eqref{eq:symmetries}.

Using the variables $\psi^+$ and $\psi^-$, the model~\eqref{twolayer} is represented as
\begin{subequations}\label{twolayer2}
\begin{align}
\begin{split}
    &Q^+_t+ \frac12\left(\{\psi^+,Q^+\} + \{\psi^-,Q^-\}\right) = 0, \\
    &Q^-_t+ \frac12\left(\{\psi^+,Q^-\} + \{\psi^-,Q^+\}\right) = 0, \\
\end{split}
\end{align}
where
\begin{align}
\begin{split}
    & Q^+ = \nn^2\psi^+ + 2\beta y,\\
    & Q^- = \nn^2\psi^- - 2F\psi^-,
\end{split}
\end{align}
\end{subequations}
are the barotropic and baroclinic potential vorticities, respectively. In the present paper, we will use both the forms of the two-layer model, i.e.~employing both the ``layered variables'' $\psi^1$ and $\psi^2$ and the ``barotropic/baroclinic variables'' $\psi^+$ and $\psi^-$. The precise usage depends on whether we study linear or nonlinear submodels of the two-layer model.

\section[Complete point symmetry group]{Complete point symmetry group}\label{sec:CompletePointSymmetryGroupQG}

The complete point symmetry group of a system of differential equations, which includes both continuous and discrete symmetries, is conventionally calculated using the \emph{direct method}. The outlines of this method are quite simple. Supposing the most general form of a point transformation in the associated space of independent and dependent variables, one expresses all involved derivatives of the new (transformed) dependent variables via the old variables, substitutes the obtained expressions into the system written in terms of the transformed variables, then excludes the derivatives which are assumed constrained due to the system (\textit{principal derivatives}) and splits with respect to the unconstrained (\textit{parametric}) derivatives. As a result, one obtains a system of determining equations for point symmetry transformations, which is nonlinear in contrast to a similar system arising under application of the infinitesimal Lie method and, therefore, is much more complicated for solving. This is why different special techniques (the implicit representation for unknown functions, the combined splitting with respect to old and new variables, inverse expression of old derivative via new ones, a mapping of the initial system by a point transformation to another system, etc.) are applied to the derivation of determining equations and their a priori simplification \cite{bihl10Cy,popo10Cy,popo10Ay,prok05Ay}.

Here we aim to use an approach similar as described in~\cite{hydo00By}, which is based on the knowledge of the Lie symmetries of a given differential equation. This method rests on the fact that any point symmetry generates an automorphism of the maximal Lie invariance algebra. By factoring out the continuous symmetries from the whole point symmetry group, the discrete symmetries of the differential equation can be determined.

The computation of the complete point symmetry group can be considerably simplified by noting that any automorphism of a Lie algebra $\mathfrak g$ leaves invariant, by definition, all megaideals of $\mathfrak g$. Recall that a \textit{megaideal} of $\mathfrak g$ is a vector subspace of $\mathfrak g$ which is invariant under any transformation from the group of automorphisms of $\mathfrak g$~\cite{popo05Ay}. Therefore, by determining megaideals of the maximal Lie invariance algebra of the given differential equation and imposing the invariance condition of these megaideals under the push-forwards of vector fields associated with the point symmetries allows to restrict the general form of possible point symmetries already \textit{before} transforming the differential equation itself. After taking into account these initial restrictions, it is usually much simpler to proceed with the splitting of the variables in the transformed differential equation as described in the first paragraph.

In this section, the approach just outlined is demonstrated for the two-layer equation in barotropic/baroclinic variables.

\begin{theorem}\label{theo:PointSymGroupOfTwoLayerEq}
The point symmetry group~$G$ of system~\eqref{twolayer2} consists of transformations of the form
\begin{align*}
 &\tilde t = \ve_1t+T_0, \quad \tilde x=\ve_1x+f(t), \quad \tilde y=\ve_2y+Y_0, \\
 &\tilde \psi^- = \ve_3\psi^- + \Psi^-_0,\quad \tilde \psi^+ = \ve_2\psi^+-2\ve_1\ve_2f_ty+g(t), 
\end{align*}
where $\ve_i =\pm1$, $i=1,2,3$; $T_0,Y_0,\Psi^-_0\in\mathbb R$ and $f$ and $g$ are arbitrary smooth functions of $t$.
\end{theorem}

\begin{proof}
Recall that the maximal Lie invariance algebra of system~\eqref{twolayer2} is the infinite dimensional algebra
$\mathfrak g=\langle\p_t,\p_y,\XX(f),\FF,\ZZ(g)\rangle$,
where $f$ and $g$ run through the set of smooth functions of $t$.
It is a solvable algebra since $\mathfrak g'=\langle\XX(f),\ZZ(g)\rangle$ and hence $\mathfrak g''=\{0\}$.
In other words, the radical~$\mathfrak r$ of~$\mathfrak g$ coincides with the entire~$\mathfrak g$.
The algebra~$\mathfrak g$ has the nontrivial center $\mathfrak z=\langle\XX(1),\FF,\ZZ(1)\rangle$.

The nil-radical of~$\mathfrak g$ is the ideal $\mathfrak n=\langle\p_y,\XX(f),\FF,\ZZ(g)\rangle.$
Indeed, this ideal is a nilpotent subalgebra of~$\mathfrak g$ since
$\mathfrak n^{(2)}=\mathfrak n'=[\mathfrak n,\mathfrak n]=\langle\ZZ(g)\rangle$ and
$\mathfrak n^{(3)}=[\mathfrak n,\mathfrak n']=0$.
A unique ideal of~$\mathfrak g$ properly containing~$\mathfrak n$ is the entire algebra~$\mathfrak g$ itself, which is not nilpotent.
This means that $\mathfrak n$ is the maximal nilpotent ideal.

In the calculations of~$G$, we use the following megaideals of $\mathfrak g$:
the entire algebra~$\mathfrak g$,
the derived algebra $\mathfrak g'$,
the nil-radical~$\mathfrak n$,
its derivative~$\mathfrak n'$,
the center $\mathfrak z$
and their proper intersections,
$\mathfrak z\cap\mathfrak g'=\langle\XX(1),\ZZ(1)\rangle$ and
$\mathfrak z\cap\mathfrak n'=\langle\ZZ(1)\rangle$.

For the general point transformation
\[
 \mathcal T\colon\quad (\tilde t,\tilde x,\tilde y, \tilde \psi^-,\tilde \psi^+) = (T,X,Y,\Psi^-,\Psi^+),
\]
where $T$, $X$, $Y$, $\Psi^-$, $\Psi^+$ are functions of $t$, $x$, $y$, $\psi^-$, $\psi^+$ with the Jacobian not equal to zero, to be a point symmetry of system~\eqref{twolayer2}, the associated push-forward $\mathcal T_*$ of vector fields must be an automorphism of $\mathfrak g$. In particular,
$\mathcal T_*\mathfrak g =\mathfrak g$,
$\mathcal T_*\mathfrak g'=\mathfrak g'$,
$\mathcal T_*\mathfrak n =\mathfrak n$,
$\mathcal T_*\mathfrak n'=\mathfrak n'$,
$\mathcal T_*\mathfrak z =\mathfrak z$,
$\mathcal T_*(\mathfrak z\cap\mathfrak g')=\mathfrak z\cap\mathfrak g'$
and
$\mathcal T_*(\mathfrak z\cap\mathfrak n')=\mathfrak z\cap\mathfrak n'$.

Investigating the restrictions on $\mathcal T$ imposed by the invariance of $\mathfrak z\cap\mathfrak n'$ under~$\mathcal T_*$, we have
\[
\mathcal T_* \ZZ(1)= 2(T_{\psi^+}\p_{\tilde t} + X_{\psi^+}\p_{\tilde x}+Y_{\psi^+}\p_{\tilde y}+\Psi^-_{\psi^+}\p_{\tilde \psi^-}+\Psi^+_{\psi^+}\p_{\tilde \psi^+})=\tilde\ZZ(a),
\]
where $a=\const$.
This equation implies that $T_{\psi^+}=X_{\psi^+}=Y_{\psi^+}= \Psi^-_{\psi^+}=0$, $\Psi^+_{\psi^+}=a=\const$,
and $a\ne0$.
Then, from the transformation of the elements of~$\mathfrak n'$ we conclude
\[
\mathcal T_* \ZZ(g) = 2ag\p_{\tilde \psi^+}=\tilde\ZZ(\tilde g^g),
\]
where $\tilde g^g$ is a smooth function of $\tilde t$ related to~$g$.
Comparing coefficients, we find that $ag(t)=\tilde g^g(T)$.
As $g$ is an arbitrary smooth function of $t$, we can set $g=t$ to obtain $2at=\tilde g^t(T)$.
Because of $a\ne0$ and hence $\tilde g^t_{\tilde t}\ne0$, this implies that $T=T(t)$ and $T_t\ne0$.

In a similar manner, the condition
\[
\mathcal T_* \XX(1) = \tilde \XX(b_1) + \tilde \ZZ (b_2),
\]
with $b_1,b_2=\const$ follows from the invariance of $\mathfrak z\cap\mathfrak g'$ with respect to $\mathcal T_*$.
It is spit into the equations $Y_x=\Psi^-_x = 0$, $X_x=b_1=\const$ and $\Psi^+_x = 2b_2=\const$, where $b_1\ne0$.
The invariance of $\mathfrak g'$ implies the condition
\[
 \mathcal T_* \XX(f) = \tilde \XX(\tilde f^f) + \ZZ(\tilde g^f),
\]
where $\tilde f^f$ and $\tilde g^f$ are smooth functions of $\tilde t$ related to~$f$.
Comparing coefficients in the last condition immediately gives the additional equations
$b_1f=\tilde f^f(T)$ and $b_2f-af_ty=-\tilde f^f_{\tilde t}(T)Y+\tilde g^f(T)$.
Taking into account the equality $\tilde f^f_{\tilde t}=b_1f_t/T_t$, we obtain that
$Y=Y^1(t)y+Y^0(t)$, where $Y^1=aT_t/b_1\ne0$ and the precise expression of~$Y^0$ is not essential for this time.

The push-forward of the remaining basis operator $\FF$ of the center~$\mathfrak z$ by~$\mathcal T$ implies
\[
 \mathcal T_* \FF = \tilde \XX(c_1) + c_2\FF + \tilde \ZZ(c_3),
\]
where again $c_1,c_2,c_3=\const$. From this condition, we can conclude that $X_{\psi^-} = c_1=\const$, $\Psi^-_{\psi^-} = c_2=\const$ and $\Psi^+_{\psi^-} = 2c_3=\const$.

It remains to investigate the restrictions imposed by the push-forwards of the basis operators $\p_t$ and $\p_y$.
The operator~$\p_t$ does not lie in the above proper megaideals. Hence, its push-forward can be represented only as a general element of~$\mathfrak g$:
\[
 \mathcal T_* \p_t = d_1\p_{\tilde t} + d_2\p_{\tilde y} + d_3\tilde \FF + \tilde \XX(\tilde f) + \tilde \ZZ(\tilde g)
\]
for real constants $d_1$, $d_2$, $d_3$ and smooth functions $\tilde f$ and $\tilde g$  of~$\tilde t$.
It is then straightforward to find $T_t=d_1=\const$, $Y_t=d_2=\const$ and thus $Y^1=\const$ and $Y^0=d_2t+d_4$, where $d_4=\const$.
Moreover, $\Psi^-_t=2d_3=\const$, $X_t=\tilde f(T)$ is a function of $t$ and $\Psi^+_t=-2\tilde f_{\tilde t}(T)\tilde y + \tilde g(T)$.

Since the operator~$\p_y$ belongs to the megaideal~$\mathfrak n$, its push-forward by~$\mathcal T$ should take the form
\[
\mathcal T_* \p_y=e_1\p_{\tilde y}+e_2\FF+\tilde \XX(\tilde f^y)+\tilde\ZZ(\tilde g^y)
\]
with $e_1,e_2=\const$ and $\tilde f^y$ and $\tilde g^y$ again being smooth functions of~$\tilde t$.
Comparing coefficients implies $Y^1=e_1$, $\Psi^-_y=2e_2$, $\Psi^+_y=-2\tilde f^y_{\tilde t}(T)\tilde y + 2\tilde g^y(T)$
and $X_y=\tilde f^y(T)$, which is a function of~$t$.
As $X_{ty}=0$, $X_y=\const$ holds. Therefore, $\tilde f^y_{\tilde t}=0$ and $\Psi^+_y$ depends only on~$t$.

Collecting all constraints obtained so far and re-denoting the involved values, we obtain the representation of the point transformations
inducing automorphisms of~$\mathfrak g$:
\begin{align}\label{eq:appendix1}
\begin{split}
 &T = T_1t  + T_0, \quad X = X_1x+X_2y + X_3\psi^- + f(t), \quad Y=Y_1y+Y_2t+Y_0, \\
 &\Psi^- = \Psi^-_1\psi^- + \Psi^-_2y + \Psi^-_3t+\Psi^-_0, \\
 &\Psi^+ = \Psi^+_1\psi^+ + \Psi^+_2x+\Psi^+_3\psi^- + \varphi(t)y + g(t),
\end{split}
\end{align}
where $T_0$, $T_1$, $X_1$, $X_2$, $X_3$, $Y_0$, $Y_1$, $Y_2$, $\Psi^-_0$, \dots, $\Psi^-_3$, $\Psi^+_1$, \dots, $\Psi^+_3$ are constants, $T_1X_1Y_1\Psi^-_1\Psi^+_1\ne0$, $X_1Y_1=\Psi^+_1T_1$, $f$~and~$g$ are smooth functions of~$t$ and $\varphi_t=-2Y_1f_{tt}$.

We now have to take into account that the transformation $\mathcal T$ is a point symmetry of system~\eqref{twolayer2}.
Therefore, we should find explicit expressions for the derivatives of $\tilde\psi^-$ and $\tilde\psi^+$
with respect to the new variables $\tilde t$, $\tilde x$, $\tilde y$.
In view of the representation~\eqref{eq:appendix1}, the transformation rules for the partial derivative operators read
\begin{align*}
 &\p_{\tilde t} = \frac{1}{T_1}\left(\DD_t - \frac{\DD_tX}{\DD_xX}\DD_x - \frac{Y_2}{Y_1}\left(\DD_y - \frac{\DD_yX}{\DD_xX}\DD_x\right)\right), \\
 &\p_{\tilde x} = \frac{1}{\DD_xX}\DD_x, \quad \p_{\tilde y} = \frac{1}{Y_1}\left(\DD_y - \frac{\DD_yX}{\DD_xX}\DD_x\right), 
\end{align*}
where $\DD_t$, $\DD_x$ and $\DD_y$ denote the operators of total differentiation with respect to $t$, $x$ and $y$, respectively.
The derivative $\psi^+_{txy}$ can arise only in the expression for $\tilde \psi^+_{\tilde t \tilde y\tilde y}$:
\[
\tilde \psi^+_{\tilde t \tilde y\tilde y} = -\frac{2}{T_1Y^2_1}\frac{\DD_yX}{\DD_xX}\psi^+_{txy} +\cdots.
\]
Since in the first equation of~\eqref{twolayer2} there is no term with $\psi^+_{txy}$, we consequently have $\DD_yX=0$, which implies $X_2=X_3=0$.
Using this result, the transformation rules for the partial derivative operators are essentially simplified:
\[
\p_{\tilde t} = \frac{1}{T_1}\left(\p_t - \frac{f_t}{X_1}\p_x - \frac{Y_2}{Y_1}\p_y\right), \quad
\p_{\tilde x} = \frac{1}{X_1}\p_x, \quad
\p_{\tilde y} = \frac{1}{Y_1}\p_y.
\]
Upon the substitution $\psi^+_{xx} = \nn^2 \psi^+-\psi^+_{yy}$, we obtain
\[
\tilde\nn^2 \tilde \psi^+_{\tilde t} = \frac{1}{T_1}\left(\frac{1}{X_1^2}\nn^2 \psi^+_t + \left(\frac{1}{Y_1^2}-\frac{1}{X_1^2}\right)\psi^+_{tyy}\right) +\cdots.
\]
Since there is no extra term $\psi^+_{tyy}$ in the first equation of~\eqref{twolayer2}, we have $X_1^2=Y_1^2=:\sigma\ne0$.

To plug the transformed variables into system~\eqref{twolayer2}, it is convenient to write down the expressions for the transformed potential vorticities $\tilde Q^-$ and $\tilde Q^+$:
\begin{align*}
 &\tilde Q^+ = \frac{\Psi^+_1}{\sigma}(Q^+ - 2\beta y) + \frac{\Psi^+_3}{\sigma}(Q^-+2F\psi^-)+ 2\beta(Y_1x+Y_2t+Y_0), \\
 &\tilde Q^- = \frac{\Psi^-_1}{\sigma}(Q^- + 2F\psi^-) - 2F(\Psi^-_1\psi^- + \Psi^-_2y+\Psi^-_3t + \Psi^-_0).
\end{align*}
The further restrictions on the transformation $\mathcal T$ can be found by splitting the resulting equations
with respect to the variables $t$, $x$, $y$, $\psi^\pm_x$,  $\psi^\pm_y$, $Q^\pm_x$ and $Q^\pm_y$.

We start with the restrictions imposed from the transformation of the second equation of system~\eqref{twolayer2}.
The term with $\psi^-_xQ^-_y$ arises after the expansion of $\tilde \psi^+_{\tilde x}\tilde Q^-_{\tilde y}+\tilde \psi^-_{\tilde x}\tilde Q^+_{\tilde y}$.
As $\Psi^-_1\ne0$, the corresponding coefficient, which is equal to $\Psi^+_3\Psi^-_1/(\sigma X_1Y_1)$, vanishes if and only if $\Psi^+_3=0$.
A term with $\psi^-_t$ is contained only in the expression for $\tilde Q^-_{\tilde t}$.
The corresponding coefficient $2FT_1^{-1}\Psi^-_1(\sigma^{-1}-1)$ also must be equal to zero, i.e., $\sigma = 1$ and hence $X_1=\ve_1=\pm1$ and $Y_1=\ve_2=\pm 1$.
Splitting in a similar way with respect to $\psi^+_x$ and $\psi^-_y$, we respectively find that $\Psi^-_2=0$ and $Y_1=\Psi^+_1$.
As we already know that $Y^1=\Psi^+_1T_1/X_1$, this implies that $T_1=X_1=\ve_1$.

From these restrictions, we can conclude the following form of the transformations:
\begin{align*}
 &T = \ve_1t+T_0, \quad X=\ve_1x+f(t), \quad Y=\ve_2y+Y_2t+Y_0, \\
 &\Psi^- = \Psi^-_1\psi^- + \Psi^-_3t+\Psi^-_0,\quad \Psi^+ = \ve_2\psi^+ + \Psi^+_2x + \varphi(t)y+g(t), \\
 &\tilde Q^- = \Psi^-_1 Q^- - 2F(\Psi^-_3t+\Psi^-_0), \quad \tilde Q^+ = \ve_2 Q^+ + 2\beta(Y_2t+Y_0).
\end{align*}
Substituting these transformations into the second equation of system~\eqref{twolayer2}, we find
\begin{align*}
 &\frac{\Psi^-_1}{\ve_1}\left(Q^-_t - \frac{f_t}{\ve_1}Q^-_x - \frac{Y_2}{\ve_2}Q^-_y\right) - \frac{2F}{\ve_1}\Psi^-_3 + \frac{1}{2\ve_1\ve_2}(\ve_2\psi^+_x+\Psi^+_2)\Psi^-_1Q^-_y -  (\ve_2\psi^+_y+\varphi)\Psi^-_1Q^-_x)\\
 &+\frac{\Psi^-_1}{2\ve_1}(\psi^-_xQ^+_y-\psi^-_yQ^+_x) = \frac{\Psi^-_1}{\ve_1}\left(Q^-_t  +\frac12(\psi^+_xQ^-_y -  \psi^+_y\Psi^-_1Q^-_x)+\frac12(\psi^-_xQ^+_y-\psi^-_yQ^+_x)\right).
\end{align*}
Splitting of this equation immediately gives $\Psi^-_3=0$, $\varphi = -2\ve_1^{-1}\ve_2f_t$ and $Y_2 = \tfrac{1}{2}\Psi^+_2$.

In a similar fashion, plugging these transformations into the first equation of system~\eqref{twolayer2} we obtain
\begin{align*}
 &\frac{\ve_2}{\ve_1}\left(Q^+_t - \frac{f_t}{\ve_1}Q^+_x - \frac{Y_2}{\ve_2}Q^+_y\right) + \frac{2\beta}{\ve_1}Y_2 + \frac{1}{2\ve_1\ve_2}((\ve_2\psi^+_x+\Psi^+_2)\ve_2Q^+_y -  (\ve_2\psi^+_y+\varphi)\ve_2 Q^+_x)\\
 &+\frac{(\Psi^-_1)^2}{2\ve_1\ve_2}(\psi^-_xQ^-_y-\psi^-_yQ^-_x) = \frac{\ve_2}{\ve_1}\left(Q^+_t + \frac12(\psi^+_xQ^+_y - \psi^+_yQ^+_x)+\frac12(\psi^-_xQ^-_y-\psi^-_yQ^-_x)\right).
\end{align*}
The symmetry condition implies $\Psi^-_1=\ve_3=\pm1$, $Y_2=0$ and consequently $\Psi^+_2=0$. This completes the proof of the theorem.
\end{proof}

\begin{remark}
The continuous transformations generated by elements of the center $\mathfrak z$ and only such transformations from the point symmetry group of system~\eqref{twolayer2} induce the identical automorphism of the algebra~$\mathfrak g$.
\end{remark}

\begin{remark}
Comparing the results presented in Theorem~\ref{theo:PointSymGroupOfTwoLayerEq} and Section~\ref{sec:symmetriesQG} leads to the conclusion that besides Lie point symmetries system~\eqref{twolayer2} admits discrete point symmetries.
The group of discrete symmetries is generated by the mirror symmetries
$(t,x,y,\psi_1,\psi_2)\mapsto (-t,-x,y,\psi_1,\psi_2)$, $(t,x,y,\psi_1,\psi_2)\mapsto (t,x,-y,-\psi_1,-\psi_2)$ and $(t,x,y,\psi_1,\psi_2)\mapsto(t,x,y,\psi_2,\psi_1)$.
Under the change of the ``barotropic/baroclinic variables'' $(\psi^+,\psi^-)$ by the ``layered variables'' $(\psi^1,\psi^2)$, these discrete symmetries are changed to the transformations mapping $(t,x,y,\psi^+,\psi^-)$ to $(-t,-x,y,\psi^+,\psi^-)$, $(t,x,-y,-\psi^+,-\psi^-)$ and $(t,x,y,\psi^+,-\psi^-)$, respectively.
They exhaust the independent discrete symmetries of system~\eqref{twolayer2} up to mutual composing and composing with continuous symmetries.
Although the proof of this claim involves cumbersome calculations, 
the discrete symmetries are not essential for further consideration and hence will be neglected here. 
At the same time, they may be important for other purposes~\cite{bihl09By,hydo00By}.
\end{remark}

\section{Inequivalent subalgebras}\label{sec:InequivalentSubalgebras}

To classify the inequivalent subalgebras of the maximal Lie invariance algebra~$\mathfrak g$ of the system~\eqref{twolayer} (resp. the system~\eqref{twolayer2}), we need the adjoint representation of the corresponding Lie group on~$\mathfrak g$. See, e.g.,~\cite{fush94Ay,olve86Ay,ovsi82Ay} for details of classification techniques. We list only the nontrivial actions associated with basis elements of~$\mathfrak g$:
\begin{align*}
    &\Ad(e^{\ve \ZZ(g)})\p_t = \p_t + \ve\ZZ(g'), & & \Ad(e^{\ve\p_t})\ZZ(g) = Z(g(t-\ve)), \\
    &\Ad(e^{\ve \XX(f)})\p_t = \p_t + \ve\XX(f'), & & \Ad(e^{\ve\p_t})\XX(f) = X(f(t-\ve)), \\
    &\Ad(e^{\ve \XX(f)})\p_y = \p_y - \ve\ZZ(f'), & & \Ad(e^{\ve\p_y})\XX(f) = \XX(f) + \ve\ZZ(f').
\end{align*}
The classification of inequivalent one-dimensional subalgebras is straightforward. An optimal set of such subalgebras reads
\begin{equation}\label{eq:algebra1}
    \mathcal A^1_1=\langle \p_t + a\p_y + b\FF\rangle, \qquad \mathcal A^1_2=\langle \p_y + \XX(f) + b\FF\rangle,\qquad \mathcal A^1_3=\langle \XX(f) + \ZZ(g)+b\FF\rangle,
\end{equation}
where $a,b \in \mathbb{R}$ and $f$ and $g$ are arbitrary functions of $t$.

The classification of the two-dimensional subalgebras yields the following list:
\begin{gather}\label{eq:algebra2}
\begin{split}
    &\mathcal A^2_1=\langle\p_t+\kappa\FF, \p_y+\XX(\nu)+\ZZ(\mu)+\rho\FF\rangle, \\
    &\mathcal A^2_2=\langle \p_t + \nu\p_y+\kappa\FF, \XX(e^{\sigma t})+\ZZ(\nu\sigma te^{\sigma t})\rangle, \sigma\ne0 \\
    &\mathcal A^2_{-1}=\langle \p_t + \nu\p_y+\kappa\FF, \ZZ(e^{\sigma t})\rangle, \\
    &\mathcal A^2_3=\langle \p_t + \nu\p_y+\kappa\FF, \XX(1)+\ZZ(\mu)+\rho\FF\rangle, \\
    &\mathcal A^2_{-2}=\langle \p_t + \nu\p_y+\kappa\FF, \ZZ(1)+\rho\FF\rangle, \quad
     \mathcal A^2_{-3}=\langle\p_t+\nu\p_y, \FF\rangle, \\
    &\mathcal A^2_4=\langle \p_y + \XX(f)+\kappa\FF, \XX(1)+\ZZ(g)+\rho\FF\rangle,\ \kappa \rho=0,  \\
    &\mathcal A^2_{-4}=\langle \p_y + \XX(f), \ZZ(g)+\FF\rangle,  \quad \mathcal A^2_{-5}=\langle \p_y + \XX(f)+\kappa\FF, \ZZ(g)\rangle, \\
    &\mathcal A^2_{-6}=\langle\XX(f^1)+\ZZ(g^1)+\kappa\FF,\XX(f^2)+\ZZ(g^2)+\rho\FF\rangle,
\end{split}
\end{gather}
where $a$, $b$, $\kappa$, $\mu$, $\nu$, $\rho$, $\sigma=\const$ and $f$, $f^1$, $f^2$, $g$, $g^1$ and $g^2$ are arbitrary functions of $t$. In the final subalgebra the tuples $(f^1,g^1,\kappa)$ and $(f^2,g^2,\rho)$ must be linearly independent for the subalgebra to really be two-dimensional. By the same reason, the parameter function $g$ is not identically equal to zero in $\mathcal A^2_{-5}$.

In fact, the above subalgebras are not single subalgebras but rather represent parameterized classes of subalgebras. This is why it would be beneficial to indicate the list of parameters in the notation of the corresponding algebras, but for the sake of brevity we omit this whenever possible. For the classes $\mathcal A^2_4$, $\mathcal A^2_{-4}$, $\mathcal A^2_5$, $\mathcal A^2_{-5}$ and $\mathcal A^2_{-6}$ the adjoint actions and linear combinations of basis elements induce equivalence relations on the corresponding sets of parameters. For example, the subalgebras $\mathcal A^2_{-4}(f,g)$ and $\mathcal A^2_{-4}(\tilde f,\tilde g)$ are equivalent if and only if $\tilde f(t)=f(t-\ve)$ and $\tilde g(t)=g(t-\ve)$ for some $\ve\in\mathbb R$.

\section{Invariant reduction with one-dimensional subalgebras}\label{sec:reduction1QG}

We present the complete list of submodels obtained under reduction using the list~\eqref{eq:algebra1}. For each submodel, we again determine its Lie symmetries, thereby seeking for hidden symmetries of the initial model. For a general discussion of the problem of hidden symmetries, see e.g.~\cite{abra06Ay}. Throughout this section $v^1$, $v^2$, $v^+$ and $v^-$ will be assumed as functions of $p$ and $q$.

\subsection{Subalgebra \texorpdfstring{$\mathcal A^1_1$}{A11}}

A suitable ansatz for reduction of~\eqref{twolayer} under the first subalgebra of~\eqref{eq:algebra1} is given by $\psi^1=v^1+bt$, $\psi^2=v^2-bt$, where $p=x$, $q=y-at$. The corresponding reduced equations read
\begin{align*}
    &aw^1_q-Fa(v^1_q-v^2_q)+2Fb - v^1_p(w^1_q+\beta+Fv^2_q)+v^1_q(w^1_p+Fv^2_p)=0, \\
    &aw^2_q+Fa(v^1_q-v^2_q)-2Fb - v^2_p(w^2_q+\beta+Fv^1_q)+v^2_q(w^2_p+Fv^1_p)=0,
\end{align*}
where $w^i=v^i_{pp}+v^i_{qq}$, $i=1,2$. The maximal Lie invariance algebra of this system is $\mathfrak g_1=\langle\p_p,\p_q,\p_{v^1},\p_{v^2}\rangle$. All operators from the algebra $\mathfrak g_1$ are induced by Lie symmetry operators of the original system~\eqref{twolayer} and hence there are no purely hidden symmetries. This is why we do not have to further reduce the above system by using the Lie method. The Lie reductions of the reduced system with respect to one-dimensional subalgebras of~$\mathfrak g_1$ are equivalent to Lie reductions of system~\eqref{twolayer} with respect to one of the listed two-dimensional subalgebras of~$\mathfrak g$. The two-dimensional reductions of system~\eqref{twolayer} are exhaustively discussed in Section~\ref{sec:reduction2QG}.

\subsection{Subalgebra \texorpdfstring{$\mathcal A^1_2$}{A12}}\label{sec:SubalgebraA12}

\noindent\textbf{Reduction using $\mathcal A^1_2$.} An ansatz associated with this subalgebra reads $\psi^1 = v^1 - \tfrac{1}{2}f'y^2+by$ and $\psi^2=v^2-\tfrac{1}{2}f'y^2-by$, where $p=x-f(t)y$ and $q=t$. It reduces system~\eqref{twolayer} to the system
\begin{align*}
\begin{split}
    &(Hv^1_{pp})_q - f'' - F(v^1_q-v^2_q) -bF(v^1_p+v^2_p)-bHv^1_{ppp} + \beta v^1_p = 0, \\
    &(Hv^2_{pp})_q - f'' + F(v^1_q-v^2_q) +bF(v^1_p+v^2_p)+bHv^2_{ppp} + \beta v^2_p = 0,
\end{split}
\end{align*}
where it was convenient to introduce the new notation \[H=1+f^2.\]

To simplify the reduced system, we use the above mentioned barotropic/baroclinic variables, which is particularly obvious for this submodel, since it is a linear system of differential equations. By introducing $w=v^1+v^2$ and $v=v^1-v^2$ we are able to rewrite the resulting system via:
\begin{align}\label{eq:redsub0}
\begin{split}
    &(Hw_{pp})_q- 2f''-bHv_{ppp}+\beta w_p  = 0, \\
    &(Hv_{pp})_q-2Fv_q-2bFw_p -2bHw_{ppp}+\beta v_p = 0.
\end{split}
\end{align}
Note that system~\eqref{eq:redsub0} may be derived directly by means of reduction of~\eqref{twolayer2} under the ansatz $\psi^+=w(p,q)-f'y^2$ and $\psi^-=v(p,q)+2by$, where $p$ and $q$ are defined as above.

The resulting system is now simplified in a way similar as presented in~\cite{bihl09Ay}. Namely, we integrate once the first equation with respect to $p$ yielding
\[
    (Hw_p)_q - 2f''p-bHv_{pp}+\beta w +h(q) = 0,
\]
where $h$ is an arbitrary function of $q=t$. Then we apply the transformations of the unknown functions
\[
    w = \hat w - 2\frac{(Hf'')'}{\beta^2} + \frac{2f''p}{\beta}- \frac{h}{\beta}, \quad v = \hat v - \frac{2bf'}{\beta}
\]
and obtain the following system:
\begin{align}\label{eq:redsub}
\begin{split}
    &(H\hat w_p)_q  + \beta \hat w - bH\hat v_{pp} = 0, \\
    &(H\hat v_{pp})_q -2F\hat v_q+\beta \hat v_p- 2b(H\hat w_{ppp}+F\hat w_p) = 0.
\end{split}
\end{align}
First of all, we determine the Lie symmetries of this system. As system~\eqref{eq:redsub} in fact is a class of systems parameterized by the arbitrary function $f=f(q)$ and the arbitrary constant~$b$, it is necessary to solve a group classification problem \cite{ovsi82Ay,popo04Ay,popo10Ay}. That is, for the complete description of Lie symmetries it is necessary to seek for possible extensions of the Lie invariance algebras for special values of the parameters~$f$ and~$b$, respectively. Recall that $\beta$ and $F$ are constant parameters which can be set to $1$, so it is not required to also take into account the classification problem with respect to~$\beta$ and~$F$.

\medskip

\noindent\textbf{Group classification of the reduced systems.} Conventionally, the first step in the procedure of group classification is the identification of the kernel $G^{\mathrm{ker}}$ of the maximal Lie invariance groups of systems from class~\eqref{eq:redsub}, i.e.\ the group which is admitted for any value of $f$ and $b$. The Lie algebra $\mathfrak g^{\mathrm{ker}}$ corresponding to $G^{\mathrm{ker}}$ can be obtained by solving the determining equations for Lie symmetries under the assumption of arbitrariness of~$f$ and~$b$. The part of the determining equations not including~$f$ and~$b$ can be immediately integrated yielding
\[
    \xi^p = ap+c, \quad \xi^q = aq+d, \quad \eta^{\hat v} =k_1\hat v+z^1(p,q), \quad \eta^{\hat w} = k_2\hat w+z^2(p,q),
\]
which are the coefficients of the most general infinitesimal generator of Lie symmetries $\xi^p\p_p+\xi^q\p_q+\eta^{\hat v}\p_{\hat v}+\eta^{\hat w}\p_{\hat w}$, where $a$, $c$, $d$, $k_1$, $k_2=\const$. The part of the determining equations explicitly including the parameters of~\eqref{eq:redsub} (the classifying part) in turn is:
\begin{align}\label{eq:classA12}
\begin{split}
    &(aq+d)H'-2aH=0, \quad (aq+d)H'' - aH'=0, \\
    &aHH'+(aq+d)HH''-(aq+d)H'^2=0, \quad b(k_1-k_2)=0,\\
    &(Hz^2_p)_q + \beta z^2 - bHz^1_{pp} = 0,\\
    &(Hz^1_{pp})_q - 2Fz^1_q + \beta z^1_p - 2b(Fz^2_p+Hz^2_{ppp}) = 0.
\end{split}
\end{align}
It is straightforward to recover system~\eqref{eq:redsub} in the two last equations of system~\eqref{eq:classA12}. For the general values of~$f$ and~$b$ splitting of the above system yields the conditions $a=0$, $d=0$ and $k_1=k_2$ and hence gives rise to the Lie invariance algebra $\mathfrak g^{\mathrm{gen}}_{f,b}$ generated by the operators
\[
    \p_p,\quad \mathcal I = \hat v\p_{\hat v} + \hat w\p_{\hat w},\quad \mathcal L (z^1,z^2) = z^1(p,q)\p_{\hat v} + z^2(p,q)\p_{\hat w},
\]
where functions $z^1$ and $z^2$ run through the set of solutions of the system~\eqref{eq:redsub} for the fixed values~$f$ and~$b$. The Lie symmetry operators $\mathcal I$ and $\mathcal L(z^1,z^2)$ arise due to linearity of~\eqref{eq:redsub}.

\looseness=-1
To investigate the problem of induced symmetries of system~\eqref{eq:redsub}, we need to consider the same problem for system~\eqref{eq:redsub0} at first. Up to linear combining, for general values of~$f$ and~$b$ the Lie symmetry operators of~\eqref{eq:redsub0} induced by operators from~$\mathfrak g$ are exhausted by the operators $\p_p$, $2\p_{\hat v}$ and $g(q)\p_{\hat w}$, which are induced by $\XX(1)$, $\FF$ and $\ZZ(g)$, respectively. Here $g$ runs through the set of smooth functions of~$q$. It is obvious that the operator $\p_y+\XX(f)+b\FF$ induces the zero operator. Additionally, if $f=\const$ the operator $\p_t$ induces $\p_q$. Under integration of the first equation of system~\eqref{eq:redsub0}, Lie symmetry transformations generated by $g(q)\p_{\hat w}$ for any fixed $g$ become equivalence transformations of the resulting system. By the above transformation to the unknown functions $\hat v$ and $\hat w$, we gauge the function $h$ arising under integration to zero and hence break the invariance of system~\eqref{eq:redsub} with respect to operators of the form $g(q)\p_{\hat w}$. This is why any operator from $\mathfrak g^{\mathrm{gen}}_{f,b}$ lying in the complement of the linear span of the operators $\p_p$, $2\p_{\hat v}$ and $g(q)\p_{\hat w}$ and, additionally, $\p_q$ in the case of $f=\const$ is a hidden symmetry of the initial system.

The kernel algebra $\mathfrak g^{\mathrm{ker}}$ of class~\eqref{eq:redsub} is generated only by the operators $\p_p$, $\mathcal I$ and $\mathcal L(1,0)$. This is because the set of generators $\mathcal L (z^1,z^2)$ is different for every representative of the class~\eqref{eq:redsub} as the form of systems from the class depends on values of $f$ and $b$. However, it is not feasible to linearly combine solutions of different systems from the class~\eqref{eq:redsub} that contradicts belonging of~$\mathcal L (z^1,z^2)$ to~$\mathfrak g^{\mathrm{ker}}$ for arbitrary values of $z^1$ and $z^2$. Therefore, under solving the group classification problem for the class~\eqref{eq:redsub} it is natural to investigate extensions with respect to $\mathfrak g^{\mathrm{gen}}_{f,b}$ rather than with respect to $\mathfrak g^{\mathrm{ker}}$. In other words, we should find all inequivalent values of the arbitrary elements~$f$ and~$b$ for which $\mathfrak g^{\mathrm{gen}}_{f,b}$ is not the maximal Lie invariance algebra of the corresponding system of the form~\eqref{eq:redsub}. Here the inequivalence is to be understood with respect to the equivalence group of the class~\eqref{eq:redsub}.

For this purpose, it is necessary to solve the classifying part~\eqref{eq:classA12} of the determining equations by taking into account for which forms of $f$ and $b$ an extension of~$\mathfrak g^{\mathrm{gen}}_{f,b}$ is admitted. It is obvious that for $b=0$, the generator $\mathcal I$ from $\mathfrak g^{\mathrm{gen}}_{f,b}$ splits into the two generators $\hat v\p_{\hat v}$ and $\hat w\p_{\hat w}$. This splitting of $\mathcal I$ corresponds to the decoupling of the two equations~\eqref{eq:redsub}.  As the remaining classifying part of the determining equations is independent of $b$, the extensions possible for different values of $H$ are essentially not affected whether or not the two equations~\eqref{eq:redsub} are coupled. In case of $b=0$ we simply consider extensions with respect to $\mathfrak g^{\mathrm{gen}}_{f,0} = \langle\p_p,\hat v\p_{\hat v},\hat w\p_{\hat w},\mathcal L (z^1,z^2) \rangle$ rather than to $\mathfrak g^{\mathrm{gen}}_{f,b}$. It is crucial to remark that from the first three equations of system~\eqref{eq:classA12} only the first equation is independent. The other two equations are its differential consequences. As the investigation of extensions must be done up to equivalence, it would be necessary to compute the equivalence group of the class~\eqref{eq:redsub}. However, it is obvious that this class admits scalings and shifts of $q$ as equivalence transformations. For this reason, we only have to distinguish between the cases $a\ne0$ and $a=0$. In the case $a\ne0$, we can set $a=1$ and $d=0$ by dividing the equation by $a$ and shifting of $q$. As a result, we have $H=\varkappa q^2$, where $\varkappa$ is a positive constant in view of the definition of~$H$, i.e., $f=\pm\sqrt{\varkappa q^2-1}$. The extension of~$\mathfrak g^{\mathrm{gen}}_{f,b}$ is then given by the basis element $p\p_p+q\p_q$ which is a hidden symmetry of the initial system~\eqref{twolayer2}. If $a=0$, we have $H=\const$ and, moreover, $H\geqslant1$, i.e., $f=\const$ and $\mathfrak g^{\mathrm{max}}_{f,b}=\mathfrak g^{\mathrm{gen}}_{f,b}+\langle\p_q\rangle$. Recall that operator $\p_q$ in this case is induced by the operator~$\p_t$. Formally, we can scale the constant value~$H$ by a uniform scaling in~$p$ and~$q$ but this is not related to the nature of the problem. 


\medskip

\noindent\textbf{The uncoupled system ($\boldsymbol{b=0}$).} We now proceed by studying the case $b=0$, which leads to a decoupling of system~\eqref{eq:redsub}:
\begin{align}\label{eq:redsubb0}
   (H\hat w_p)_q + \beta \hat w = 0, \quad (H\hat v_{pp})_q -2F\hat v_q+\beta \hat v_p = 0.
\end{align}
The change of variables
\[
    \bar p = p, \quad \bar q = \int \frac{\ddd q}{H(q)}, \quad \bar v = H(q)\hat v, \quad \bar w = H(q)\hat w,
\]
allows to transform this system to
\begin{align}\label{eq:redsublinear}
   \bar w_{\bar p\bar q} + \beta \bar w = 0, \quad \bar v_{\bar p \bar p\bar q} -2F(H^{-1}\bar v)_{\bar q}+\beta \bar v_{\bar p} = 0.
\end{align}
Analogously to~\cite{bihl09Ay}, in the first equation we recover the Klein--Gordon equation in light-cone variables. For this system, we only have one possibility for Lie reduction in the general case~of~$H$. Namely, we can reduce~\eqref{eq:redsubb0} by using the subalgebra $\langle\p_{\bar p}+\lambda_1\bar v\p_{\bar v} + \lambda_2\bar w\p_{\bar w}\rangle$, where we take into account that for the decoupled case $b=0$ the generator $\hat v\p_{\hat v} + \hat w\p_{\hat w}$ of~\eqref{eq:redsub} splits into the two single operators $\hat v\p_{\hat v}$ and $\hat w\p_{\hat w}$ (note that $\hat v\p_{\hat v} = \bar v\p_{\bar v}$ and $\hat w\p_{\hat w}= \bar w\p_{\bar w}$). The ansatz for reduction reads $\bar v = \tilde v(\bar q)e^{\lambda_1 \bar p}$ and $\bar w = \tilde w(\bar q)e^{\lambda_2 \bar p}$. It leads to the following system of ordinary differential equations
\[
    \lambda_2 \tilde w_{\bar q} + \beta \tilde w = 0, \quad \lambda_1^2 \tilde v_{\bar q} - 2F(H^{-1}\tilde v)_{\bar q} + \beta\lambda_1\tilde v = 0.
\]
The integration of the resulting system, the substitution of the obtained solution to the ansatz and the inverse change of variables yields
\[
 \hat w = \frac{c_1}{H}\exp\left(\lambda_2 p-\frac{\beta}{\lambda_2}\int \frac{\ddd q}{H}\right),\quad \hat v = \frac{c_2}{H}\exp\left(\lambda_1 p + \int\frac{2H_{q} F+\beta\lambda_1H}{2F-\lambda_1^2H}\frac{\ddd q}{H}\right).
\]
Recall that $H=1+f^2$. For the values $f=-l/k$, where the constants $k$ and $l$ are wave numbers and $\lambda_1=\lambda_2=ik$, $i^2=-1$, this solution reduces to the well-known Rossby waves in the two-layer model. The solution for the barotropic mode $\hat w$ describes a single barotropic Rossby wave, which is independent of the vertical structure of the two-layer setting. In turn, the solution for $\hat v$ describes the evolution of the first baroclinic mode and explicitly depends on the vertical layer structure via the parameter dependency on $F$, which accounts for the density difference between the layers. Note that for a more general choice of the function $f$, this solution allows to derive series of wave solutions in a similar way as it was possible for the barotropic vorticity equation~\cite{bihl09Ay}. However, for $f\ne\const$ there arises an additional term proportional to $y^2$ in the solution of $\psi^+$, which may violate the boundary conditions. Although such a violation of the boundaries does not exist for the solution of $\psi^-$, the global (i.e.~large-scale) realization of generalized Rossby waves is at once limited due to this restriction.

The case $b=0$ shows that the classical Rossby wave solution can be recovered in two steps: Firstly reducing with respect to the operator $\p_y + \XX(f) + b\FF$ and secondly performing reduction to a system of ordinary differential equations using a hidden symmetry of the submodel received in the first step.

Reduction with respect to the additional symmetries for special values $f(q)=\const$ and $f(q)=\pm\sqrt{\varkappa q^2-1}$ will not be considered here, because for a decoupled system it is be better to construct exact solutions of each equation separately and then compose them to a solution of the entire system. In our particular case, the single Klein--Gordon equation has a wider maximal Lie invariance algebra than system~\eqref{eq:redsublinear}, given by $\langle\p_{\bar p}, \ \p_{\bar q}, \ \bar p\p_{\bar p} - \bar q\p_{\bar q}, \ \bar w\p_{\bar w}, \ z(\bar p, \bar q)\p_{\bar w}\rangle$, where $g$ runs through the set of solutions of the Klein--Gordon equation. That is, we have more possibilities for finding exact solutions by Lie methods. Fortunately, it is not necessary to do this in view of large classes of exact solutions already known for the Klein--Gordon equation~\cite{poly02Ay}.

For the split system~\eqref{eq:redsublinear}, it remains to determine the Lie symmetries and perform Lie reductions of the second equation,
\begin{equation}\label{eq:sub}
    \bar v_{\bar p \bar p\bar q} -2F(A\bar v)_{\bar q}+\beta\bar v_{\bar p} = 0,
\end{equation}
where $A=H^{-1}>0$. The determining equations for the coefficients of the Lie symmetry operator $Q = \xi^{\bar p}\p_{\bar p} + \xi^{\bar p}\p_{\bar p}+ \eta\p_{\bar v}$ of Eq.~\eqref{eq:sub} not involving $A$ can be integrated to give
\[
 \xi^{\bar p} = -a\bar{p}+c, \quad \xi^{\bar q} = a{\bar q}+d, \quad \eta =k\bar v+z(\bar p,\bar q).
\]
The remaining classifying part of the determining equations reads
\begin{equation}\label{eq:classifyingpart}
 (a\bar q + d)A_{\bar q} - 2 a A = 0, \quad aA A_{\bar q} - (a\bar q+d)A_{\bar q}^2 + (a\bar q+d)AA_{\bar q\bar q} = 0.
\end{equation}
Again, the second equation is a differential consequence of the first equation. For general $A$, splitting of~\eqref{eq:classifyingpart} leads to the two essential Lie symmetry generators $\p_{\bar p}$ and $\bar v\p_{\bar v}$ together with the linearity operators $\mathcal L(z)=z(\bar p,\bar q)\p_{\bar v}$, where $z$ runs through the set of solutions of~\eqref{eq:sub}.

Upon looking for Lie symmetry extensions, we should consider the cases where $(a,d)\ne(0,0)$. We then distinguish two cases. (i) $a\ne0$. We scale $a=1$ and shift $q$ to set $d=0$ and obtain the value $A=\varkappa q^2$ with the additional Lie symmetry generator $\bar q\p_{\bar q} - \bar p \p_{\bar p}$, where $\varkappa$ is a nonvanishing constant. (ii) $a=0$, $d\ne0$. In this case we find $A=\const$ and the additional Lie symmetry generator $\p_{\bar q}$. The first operator again is a hidden symmetry, while the second generator is induced by $\p_t$. Note that the case of $A=0$ would yield wider symmetry extensions but cannot be realized in the present case since by definition $A\ne0$. Both the cases of extensions correspond to those above for arbitrary~$b$ by taking into account the change of variables leading to system~\eqref{eq:redsublinear}.

We now consider the Lie reductions of Eq.~\eqref{eq:sub}. For general $A$, the only nontrivial possibility for reduction is given by $\p_{\bar p} + \lambda \bar v\p_{\bar v}$, which was already considered above. The solution is
\[
 \bar v = c\exp\left(\lambda \bar p -\int\frac{\beta\lambda - 2FA_{\bar q}}{\lambda^2 - 2FA}\ddd \bar q\right),
\]
which can be combined with arbitrary solutions of the Klein--Gordon equation to yield a solution of the decoupled system~\eqref{eq:redsublinear}. It remains to study Lie reductions involving symmetry extensions.

In the first case we have $A=\varkappa q^2$. The maximal Lie invariance algebra of the equation~\eqref{eq:sub} with $A=\varkappa q^2$ is given by $\langle\bar q\p_{\bar q}-\bar p\p_{\bar p}, \p_{\bar p},\bar v\p_{\bar v}, \mathcal L(z)\rangle$. One more inequivalent one-dimensional subalgebra of this algebra is appropriate for Lie reduction. It reads $\langle\bar q\p_{\bar q} - \bar p \p_{\bar p} + \lambda \bar v\p_{\bar v}\rangle$. An ansatz for the reduction is $\bar v= \tilde v(r)q^\lambda$, where $r=\bar p\bar q$. The corresponding reduced equation is 
\[
 r \tilde v_{rrr} + (\lambda+2)\tilde v_{rr} + (\beta-2\varkappa Fr)\tilde v_r - 2\varkappa F(\lambda+2) \tilde v = 0.
\]
For $\lambda=-2$, this equation is solved in terms of Whittaker functions $M_{m,n}(r)$, $W_{m,n}(r)$:
\[\textstyle 
\tilde v = C_1\int M_{m,n}(\sqrt{8\varkappa F}r)\,\ddd r  + C_2\int W_{m,n}(\sqrt{8\varkappa F}r)\,\ddd r + C_3,
\]
where $m = \beta(8\varkappa F)^{-1/2}$ and $n = 1/2$. Moreover, for $\lambda=-k$, $k\in\mathbb{N}$, this equation admits polynomial solutions.

In the second case of extension, we have $A=\const$. Then, the maximal Lie invariance algebra coincides with $\langle\p_{\bar q}, \p_{\bar p},\bar v\p_{\bar v}, \mathcal L(z)\rangle$. Again, one more inequivalent one-dimensional subalgebra can be used to carry out Lie reduction, which is $\langle\p_{\bar q} + \kappa \p_{\bar p}+ \lambda\bar v\p_{\bar v}\rangle$. An appropriate ansatz for reduction is $\bar v = \tilde v(r)e^{\lambda\bar q}$, where $r=\bar p - \kappa \bar q$. Plugging this ansatz into~\eqref{eq:sub}, we find
\[
 \kappa\tilde v_{rrr} - \lambda\tilde v_{rr} - (2AF\kappa +\beta)\tilde v_{r} +2AF\lambda \tilde v = 0.
\]
This is a linear third-order ordinary differential equation with constant coefficients and thus can be solved by standard methods. In particular, as $\kappa$, $\lambda$ and $A$ are arbitrary constants, we can determine them upon prescribing a solution of the associated characteristic equation. This allows one to generate wide classes of solutions with rather different type, such as e.g.\ periodic wave solutions.

Since~\eqref{eq:sub} is a linear partial differential equation, using the same symmetries and the method of extended Lie reduction described in Appendix~\ref{sec:extendedsetofsolutionsQG} we can obtain much wider families of exact solutions for this equation, cf.\ the example in this appendix.

\medskip

\noindent\textbf{The coupled system $\boldsymbol{(b\ne0)}$.} For the coupled case $b\ne0$ the Lie reduction of system~\eqref{eq:redsub} is quite similar to the case $b=0$. For the sake of completeness, we list here all the reduced models that are possible for different values of $H$. 

For general $H$, the only possibility for reduction is due to the generator $Q=\p_p+\lambda\mathcal I$. The corresponding reduction ansatz is $\hat v = \bar v(q)e^{\lambda p}$, $\hat w = \bar w(q)e^{\lambda p}$. Plugging the ansatz into system~\eqref{eq:redsub}, one obtains
\begin{gather*}
\lambda H \bar w_q + (\lambda H_q+\beta) \bar w - b\lambda^2 H\bar v  = 0, \\
(\lambda^2H - 2F)\bar v_q + (\lambda^2H_q+\beta\lambda)\bar v - 2b\lambda(\lambda^2+F)\bar w  = 0.
\end{gather*}

For $H=\varkappa q^2$, the additional inequivalent reduction using $Q=p\p_p + q\p_q + \lambda \mathcal I$ is possible. Utilizing the ansatz $\hat v = \bar v(r)q^\lambda$, $\hat w = \bar w(r)q^\lambda$, where $r=pq^{-1}$, leads to the system
\begin{align*}
 &\varkappa(\lambda+1)\bar w_r - \varkappa r\bar w_{rr} + \beta\bar w - b\varkappa\bar v_{rr}=0, \\
 &\varkappa r\bar v_{rrr} - \lambda\varkappa\bar v_{rr} + 2F(\lambda\bar v - r\bar v_r) - \beta\bar v_r + 2b(\varkappa\bar w_{rrr}+F\bar w_r) = 0.
\end{align*}

For $H=\const$ we can also reduce~\eqref{eq:redsub} using $Q=\p_q + \kappa\p_p + \lambda \mathcal I$. The ansatz for reduction is $\hat v = \bar v(r)e^{\lambda q}$, $\hat w = \bar w(r)e^{\lambda q}$, where $r=p-\kappa q$. The resulting submodel is
\begin{align*}
 &H(\kappa \bar w_{rr} - \lambda \bar w_r) - \beta \bar w + b H\bar v_{rr} = 0, \\
 &H(\kappa\bar v_{rrr} - \lambda \bar v_{rr})  -(2F\kappa +\beta)\bar v_r +2F\lambda \bar v + 2b(H\bar w_{rrr} + F\bar w_r) = 0.
\end{align*}

All of the above reduced systems are linear systems of ordinary differential equations and, moreover, the last system has constant coefficients. As in the decoupled case $b=0$, for finding more exact solutions we could also apply the method of extended Lie reduction to~\eqref{eq:redsub} with $b\ne0$.

\subsection{Subalgebra \texorpdfstring{$\mathcal A^1_3$}{A13}}

\noindent \textbf{Reduction using $\mathcal A^1_3$.} For this subalgebra, it is convenient to start with barotropic/baroclinic variables from the beginning. Since in the case $f=0$ no Lie reduction is possible, we assume that $f\ne0$. An appropriate ansatz for reduction then reads
\[
    \psi^+=v^+-2\frac{f'y-g}{f}x,\quad \psi^-=v^-+2\frac{b}{f}x,
\]
where $p=y$ and $q=t$. Plugging this ansatz into~\eqref{twolayer2} gives reduction to the  system:
\begin{subequations}\label{eq:redA130}
\begin{gather}
v^+_{ppq}-\frac{f'p-g}{f}(v^+_{ppp}+2\beta)+\frac{b}{f}v^-_{ppp} = 0,\nonumber  \\
v^-_{ppq}-2Fv^-_q -\frac{f'p-g}{f}(v^-_{ppp}-2Fv^-_p)+\frac{b}{f}(v^+_{ppp}+2Fv^+_p+2\beta)=0,
\end{gather}
where it can be seen that the coupling between the barotropic and baroclinic parts is again provided only due to the existence of generator $\FF$. To solve this system, we integrate the first equation twice with respect to $p$ to yield
\begin{equation}\label{eq:redA130b}
    v^+_q -\frac{f'p-g}{f}v^+_p +2\frac{f'}{f}v^++\frac{b}{f}v^-_p - \frac13\frac{f'}{f}\beta p^3 + \frac{\beta g}{f}p^2+h^1(q)p+h^0(q)=0,
\end{equation}
\end{subequations}
where $h^1$ and $h^0$ are arbitrary smooth functions of $q$. By means of the change of unknown functions
\[
    v^+ = \hat v^+ + \gamma^2p^3+\gamma^1(q)p + \gamma^0(q), \quad v^-=\hat v^- + \delta^2(q)p^2+\delta^1(q)p + \delta^0(q),
\]
where
\begin{align*}
    &\gamma^2 = -\frac{\beta}{3}, && \delta^2 = -b\beta f^2\int\frac{1}{f^3}\ddd q, \\
    &\gamma^1=-\frac{1}{f}\int(2b\delta^2+fh^1)\ddd q, && \delta^1 = -2f\int\frac{g\delta^2}{f^2}\ddd q, \\
    &\gamma^0=-\frac{1}{f^2}\int f(g\gamma^1 + b\delta^1+fh^0)\ddd q, && \delta^0=\frac{1}{F}\int\frac{bF\gamma^1-gF\delta^1+f\delta^2_q}{f}\ddd q,
\end{align*}
we are able to reduce system~\eqref{eq:redA130} to the corresponding homogeneous form:
\begin{align}\label{eq:redA13}
\begin{split}
    &\hat v^+_q -\frac{f'p-g}{f}\hat v^+_p +2\frac{f'}{f}\hat v^++\frac{b}{f}\hat v^-_p =0, \\
    &\hat v^-_{ppq}-2F\hat v^-_q -\frac{f'p-g}{f}(\hat v^-_{ppp}-2F\hat v^-_p)+\frac{b}{f}(\hat v^+_{ppp}+2F\hat v^+_p)=0.
\end{split}
\end{align}
System~\eqref{eq:redA13} is simplified further using the transformation
\[
    \tilde p = f(q)p - \int g(q)\, \ddd q,\quad \tilde q = q,\quad \tilde v^+ = \hat v^+,\quad \tilde v^- = \hat v^-.
\]
In the new variables, system~\eqref{eq:redA13} becomes
\begin{align}\label{eq:redA132}
\begin{split}
    &(f^2\tilde v^+)_{\tilde q} + b(f^2\tilde v^-)_{\tilde p}=0, \\
    &(f^2\tilde v^-_{\tilde p\tilde p}-2F\tilde v^-)_{\tilde q} + b(f^2\tilde v^+_{\tilde p\tilde p} + 2F\tilde v^+)_{\tilde p}=0.
\end{split}
\end{align}
The first equation of~\eqref{eq:redA132} can be used to introduce a potential variable, via $V_{\tilde p}=f^2\tilde v^+$ and $V_{\tilde q} = -bf^2\tilde v^-$. Upon introducing $\bar q = \int f^2\ddd \tilde q$, the second equation then becomes
\begin{equation}\label{eq:redA13pot}
 f^2(f^2V_{\tilde p\tilde p\bar q})_{\bar q} - 2Ff^2V_{\bar q\bar q} - b^2\left(V_{\tilde p\tilde p\tilde p\tilde p} + \frac{2F}{f^2}V_{\tilde p\tilde p}\right)=0.
\end{equation}

\medskip

\noindent\textbf{The decoupled system ($\boldsymbol{b=0}$).} The general solution of the decoupled system~\eqref{eq:redA132} is
\[
 \hat v^+ = \frac{\zeta^1(\tilde p)}{f^2},\quad \hat v^- = \zeta^2(\tilde p) + \vartheta^1(q)e^{\sqrt{2F}p} + \vartheta^2(q)e^{-\sqrt{2F}p},
\]
where $\zeta^1$ and $\zeta^2$ are arbitrary functions of $\tilde p = f(q)p - \int g(q)\, \ddd q$ and $\vartheta^1$ and $\vartheta^2$ are arbitrary functions of~$q$.

\medskip

\noindent\textbf{The coupled system ($\boldsymbol{b\ne0}$).} As we found the general solution for the decoupled case of systems of the form~\eqref{eq:redA132}, we should not study further Lie reductions for this case. Hence we carry out Lie reductions solely for the case of $b\ne0$. For this reason, it is necessary to solve the group classification problem for the class~\eqref{eq:redA132} with $f$, $g$ and~$b$ as arbitrary elements under the assumption that $b\ne0$. As the class~\eqref{eq:redA132} consists of linear systems, the procedure of its group classification should be done as described in the previous section. The general form of Lie symmetry operators of a system from the class~\eqref{eq:redA132} is $\xi^{\tilde p}\p_{\tilde p}+\xi^{\tilde q}\p_{\tilde q}+\eta^+\p_{\tilde v^+}+\eta^-\p_{\tilde v^-}$, where the coefficients are functions of $\tilde p$, $\tilde q$, $\tilde v^+$ and $\tilde v^-$. The solution of the corresponding determining equations gives
\[
    \xi^{\tilde p} = a\tilde p + c,\quad \xi^{\tilde q} = a\tilde q+d, \quad \eta^+ = k\tilde v^+ + z^1(\tilde q,\tilde p), \quad \eta^- = k\tilde v^- + z^2(\tilde q,\tilde p).
\]
Additionally there is the single classifying equation $(aq+d)f_q-af=0$ (up to its differential consequences). For the general value of the arbitrary element~$f$ we obtain $a=d=0$. The maximal Lie invariance algebra in this case reads
\[
  \mathfrak g^{\mathrm{gen}}_f=\langle\p_{\tilde p},\ \mathcal I = \tilde v^+\p_{\tilde v^+} + \tilde v^-\p_{\tilde v^-},\ \mathcal L(\tilde z^1,\tilde z^2) = z^1(\tilde p,\tilde q)\p_{\tilde v^+}+z^2(\tilde p,\tilde q)\p_{\tilde v^-}\rangle,
\]
where $(z^1,z^2)$ run through the set of solutions of system~\eqref{eq:redA132}. Inequivalent extensions of the algebra~$\mathfrak g^{\mathrm{gen}}_f$ only exist for $f=\const$ and $f=\varkappa q$. The first case leads to the extension of~$\mathfrak g^{\mathrm{gen}}_f$ by $\p_{\tilde q}$, the second case gives the extension of~$\mathfrak g^{\mathrm{gen}}_f$ by $\tilde p\p_{\tilde p} + \tilde q\p_{\tilde q}$.

To study the problem of induced Lie symmetries of system~\eqref{eq:redA132} with  $b\ne0$, we need to consider the same problem for the initial reduced system at first. Up to linear combining, for general values of~$f$ and~$g$ the Lie symmetry operators of~\eqref{eq:redsub0} induced by operators from~$\mathfrak g$ are exhausted by the operators 
\[
\tilde\XX(\tilde f)=-2\left(\tilde f_qp-\frac{f_q}f\tilde fp+\frac gf\tilde f\right)\p_{v^+}-2\frac bf\tilde f\p_{v^-},\quad  
\tilde\FF=2\p_{v^-}, \quad \tilde\ZZ(\tilde g)=\tilde g(q)\p_{v^+}, 
\]
which are induced by $\XX(\tilde f)$, $\FF$ and $\ZZ(\tilde g)$, respectively. Here $\tilde f$ and~$\tilde g$ run through the set of smooth functions of~$q=t$. Additionally, if $f=\const$ the operator $\p_y$ induces $\p_p$ and if also $g=\const$ the operator $\p_t$ induces $\p_q$. 
Lie symmetry transformations generated by $\tilde\XX(\tilde f)+\tilde\ZZ(\tilde g)$ with any fixed $\tilde f$ and~$\tilde g$ are equivalence transformations of the system~\eqref{eq:redA130}. By these transformations, we gauge the functions $h^0$ and~$h^1$ in~\eqref{eq:redA130b} to zero and hence break the invariance with respect to the operators $\tilde\XX(\tilde f)$ and $\tilde\ZZ(\tilde g)$ with the general values of~$\tilde f$ and~$\tilde g$. Summing up, we can say that almost all operators from $\mathfrak g^{\mathrm{gen}}_f$ and its extensions are hidden symmetries of the initial system~\eqref{twolayer}.

It was shown before that by introducing a potential variable, a system of the form~\eqref{eq:redA132} can be converted into a single equation in the conserved form~\eqref{eq:redA13pot}. Since such an equation may have additional symmetries compared to the original system, we also carry out the group classification of the class~\eqref{eq:redA13pot} (again under the assumption that $b\ne0$). Solving the determining equations for the coefficients of a Lie symmetry operator $\xi^{\tilde p}\p_{\tilde p}+\xi^{\bar q}\p_{\bar q}+\eta^V\p_V$ of~\eqref{eq:redA13pot} gives
\[
    \xi^{\tilde p} = a\tilde p + c,\quad \xi^{\bar q} = 3a\bar q+d, \quad \eta^V = \alpha V+z(\tilde p,\bar q),
\]
together with the classifying equation $(3a\bar q+d)f'-af = 0$. For arbitrary $f$, this immediately implies that $a=d=0$ and hence gives rise to the maximal Lie invariance algebra $\langle\p_{\tilde p},\, V\p_V,\, z(\tilde p,\bar q)\p_V\rangle$, where $g$ is an arbitrary solution of~\eqref{eq:redA13pot}. There are two inequivalent extensions of this algebra, depending on either $a\ne0$, $d=0$ or $a=0$, $d\ne 0$. Since scalings and shifts in $q$ are equivalence transformations, we may first set $a=1/3$, $d=0$ with $f=\varkappa\sqrt[3]{\bar q}$. The additional generator then reads $\tilde p\p_{\tilde p}+3\bar q\p_{\bar q}$. The second case of extension is given by $a=0$ and $d=1$, leading to $f=\const$ with the corresponding additional generator $\p_{\tilde q}$.

Comparing this result with the classification of system~\eqref{eq:redA132}, we see that all Lie symmetries of the potential equation~\eqref{eq:redA13pot} are induced by Lie symmetries of~\eqref{eq:redA132} (note the change of the variable $\tilde q$ in the potential case). That is, for the reduced system~\eqref{eq:redA132} no purely potential symmetries associated with the potential equation~\eqref{eq:redA13pot} exist.

Now that we have investigated all symmetry extensions for particular values of $f(q)$, it remains to present the corresponding Lie reductions of~\eqref{eq:redA132}. The only feasible way for reduction for general $f$ is due to the operator $\p_{\tilde p} + \lambda\mathcal I$. The ansatz for reduction then is $\tilde v^+=\bar v^+(\tilde q)e^{\lambda \tilde p}$ and $\tilde v^-=\bar v^-(\tilde q)e^{\lambda \tilde p}$. The system of reduced equation reads
\begin{gather*}
    u_{\tilde q} + bf^2\lambda \bar v^- = 0, \quad (f^2\lambda^2\bar v^- - 2F\bar v^-)_{\tilde q} + b\lambda\left(\lambda^2 + \frac{2F}{f^2}\right)u=0,
\end{gather*}
where $u=f^2\bar v^+$. The case of $\lambda =0$ gives only a trivial solution and will not be considered here. For $\lambda \ne0$ it is possible to solve the first equation for $\bar v^-$. Plugging the corresponding expression into the second equation, the following homogeneous second order ordinary differential equation with variable coefficients for~$u$ is obtained
\[
u_{\tilde q\tilde q}+\frac{4F}{\lambda^2f^2-2F}\frac{f_{\tilde q}}fu_{\tilde q}-b^2\lambda^2\frac{\lambda^2f^2+2F}{\lambda^2f^2-2F}u = 0.
\]

Two more Lie reductions are possible for the particular values $f=\varkappa=\const$ and $f=\varkappa q$. In the first case, the maximal Lie invariance algebra reads $\langle\p_{\tilde q}, \p_{\tilde p}, \mathcal I, \mathcal L(\tilde z^1,\tilde z^2) \rangle$. We aim to reduce~\eqref{eq:redA132} using the generator $Q=\p_{\tilde q}+\kappa\p_{\tilde p} + \lambda\mathcal I$. The ansatz for reduction is $\tilde v^+ = \bar v^+(r)e^{\lambda \tilde q}$ and $\tilde v^- = \bar v^-(r)e^{\lambda \tilde q}$, where $r=\tilde p-\kappa\tilde q$. Plugging this ansatz into~\eqref{eq:redA132}, we obtain
\begin{gather*}
\kappa  \bar v^+_r - \lambda \bar v^+ - b\bar v^-_r = 0, \\
\varkappa^2(\kappa \bar v^-_{rrr} - \lambda\bar v^-_{rr}) - 2F(\kappa \bar v^-_r - \lambda \bar v^-) - b(\varkappa^2\bar v^+_{rrr} + 2Fv^+_{r}) = 0.
\end{gather*}
This is a coupled system of first and third order ordinary differential equations with constant coefficients and thus can be solved explicitly using standard techniques.

In the second case, the maximal Lie invariance algebra is given by 
$\langle\tilde q\p_{\tilde q}+\tilde p \p_{\tilde p}, \p_{\tilde p}, \mathcal I, \mathcal L(\tilde z^1,\tilde z^2) \rangle.$ 
One-dimensional reduction is feasible using the generator $Q=\tilde q\p_{\tilde q}+\tilde p \p_{\tilde p}+\lambda\mathcal I$. The corresponding ansatz for reduction is $\tilde v^+=\bar v^+(r)q^\lambda$ and $\tilde v^- = \bar v^-(r)q^\lambda$, where $r=pq^{-1}$. This ansatz leads to the following reduction of system~\eqref{eq:redA132}:
\begin{gather*}
  r\bar v^+_r-(\lambda+2)\bar v^+ - b\bar v^-_r = 0, \\
  r\bar v^-_{rrr} - (2\varkappa^2+\varkappa(\lambda -2))\bar v^-_{rr} - 2F(r\bar v^-_r - \lambda v^-) - b(\varkappa^2\bar v^+_{rrr} + 2F\bar v^+_r) = 0.
\end{gather*}

\begin{remark}
Substituting back the solutions of the reduced equations into the ansatz for the original unknown functions $\psi^+$ and $\psi^-$, we find that all solutions have time-dependent polynomial parts in $x$ and $y$. Since this part is not compatible with typical boundaries of the two-layer equations (as discussed in Section~\ref{sec:boundaryQG}), only solutions for the restricted case of $b=g=0$, $f=\const$ might give candidate solutions that could be realized in the framework of geophysical fluid dynamics. However, since the solution of the decoupled case $b=0$ for $g=0$, $f=\const$ is rather trivial, it is not too interesting from the physical point of view.
\end{remark}

\section{Invariant reduction with two-dimensional subalgebras}\label{sec:reduction2QG}

Now we carry out Lie reductions of~\eqref{twolayer} or~\eqref{twolayer2} to systems of ordinary differential equations using the optimal set of two-dimensional inequivalent subalgebras. Note that not all subalgebras give rise to classical group-invariant reductions. Namely, all the subalgebras from the list~\eqref{eq:algebra2} with negative subscripts cannot be used for the construction of well-defined ansatzes for the dependent variables. However, these algebras would be well suited for the construction of partially-invariant solutions, but we do not pursue this idea further in the present paper.

\subsection{Subalgebra \texorpdfstring{$\mathcal A^2_1$}{A21}}

An ansatz for group-invariant reduction using this subalgebra is $\psi^1=v^1(p)+\kappa t+(\mu+\rho)y$ and $\psi^2=v^2(p)-\kappa t+(\mu-\rho)y$, where $p=x-\nu y$. System~\eqref{twolayer} is reduced by this ansatz to
\begin{align}\label{eq:redA21}
\begin{split}
    -&(\rho+\mu)(1+\nu^2)v^1_{ppp}+ F\mu(v^1_p-v^2_p) -  F\rho(v^1_p+v^2_p) -2 F\kappa+\beta v^1_p = 0, \\
     &(\rho- \mu)(1+\nu^2)v^2_{ppp}-  F\mu(v^1_p-v^2_p) +  F\rho(v^1_p+v^2_p) +2 F\kappa + \beta v^2_p = 0.
\end{split}
\end{align}
Integrating once with respect to $p$, the above system becomes an inhomogeneous system of two second order linear ordinary differential equations with constants coefficients, provided we screen out the singular cases of $\rho=\mu=0$, $\rho=\mu\ge0$, $\rho=\mu <0$, $\rho=-\mu\ge0$ and $\rho=-\mu<0$. The solution of this system in the nonsingular case is straightforward but a bit lengthy and is therefore omitted here. Hence, we will focus on the listed singular cases only.

\medskip

\noindent\textbf{Case} $\rho=\mu=0$. The general solution for this case is
\[
    v^1 = \frac{2 F\kappa}{\beta}p + c_1, \qquad v^2 = -\frac{2 F\kappa}{\beta}p + c_2,
\]
which in the original variables $(\psi^1,\psi^2)$ gives a solution linear in $(x,y)$ and thus represents a constant wind field in both layers.

\medskip

\noindent\textbf{Case} $\rho=\mu\ge0$. This case leads to the semi-coupled system
\[
    2\mu(1+\nu^2)v^1_{ppp} + 2 F\mu v^2_p+2 F\kappa - \beta v^1_p = 0,\quad 2 F\mu v^2_p + 2 F\kappa + \beta v^2_p = 0,
\]
which is integrated to yield
\begin{align*}
    &v^1 = c_1\exp\left(\sqrt{\frac{\beta}{2\mu(1+\nu^2)}}p\right)+c_2\exp\left(-\sqrt{\frac{\beta}{2\mu(1+\nu^2)}}p\right)+ \frac{2F\kappa p}{2F\mu +\beta} + c_3, \\
    &v^2 = -\frac{2F\kappa p}{2F\mu +\beta} + c_4.
\end{align*}
In the original variables, this represents a simple exponential solution and is thus unphysical.

\medskip

\noindent\textbf{Case} $\rho=\mu<0$. The general solution in this case is
\begin{align*}
        &v^1 = c_1\cos\left(\sqrt{\frac{\beta}{2|\mu|(1+\nu^2)}}p\right)+c_2\sin\left(\sqrt{\frac{\beta}{2|\mu|(1+\nu^2)}}p\right)-\frac{2F\kappa p}{2F|\mu|-\beta}+c_3, \\
        & v^2 = \frac{2F\kappa  p}{2F|\mu|-\beta}+c_4,
\end{align*}
which in the original variables represents a single (stationary) Rossby-wave in the upper layer and a constant velocity field in the lower layer, see Figure~\ref{fig:RossbyWave}. This is a typical situation, in the study of baroclinic instability: An initial disturbance in the middle of the troposphere may start to grow while the lower part of the troposphere does not exhibit any peculiarities. It is not before the upper Rossby-wave starts unstable growth, that the lower layer also begins to show some wave-like disturbances, which subsequently may lead to the onset of cyclogenesis. Hence, the above exact solution may characterize the situation at the onset of baroclinic instability, where due to external forcing a Rossby wave in the upper layer is generated, while the wind field in the lower layer is still unaffected.

\begin{figure}[!htp]
\centering
  \includegraphics[scale=0.7]{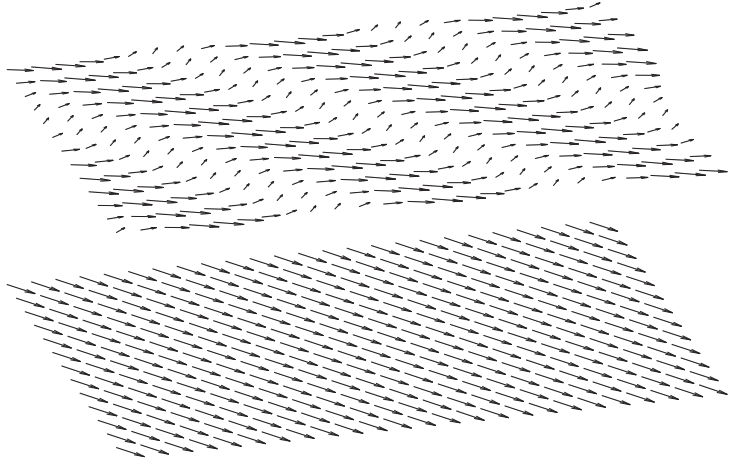}
  \caption{Rossby wave solution in the quasi-geostrophic two-layer model.}\label{fig:RossbyWave}
\end{figure}

\medskip

\noindent The cases $\rho=-\mu\ge0$ and $\rho=-\mu<0$ give the same solutions as in the two previous cases, except for interchanging the two layers, i.e.\ $v^1\leftrightarrow v^2$ and permuting the sign of the linear in $p$ term.

\subsection{Subalgebra \texorpdfstring{$\mathcal A^2_2$}{A22}}

It is convenient to use at once the quasi-geostrophic equations in terms of barotropic/baroclinic variables to perform the reduction. The ansatz we choose is: $\psi^+=v^1(p)-2\sigma px$ and $\psi^-=v^2(p)+2\kappa t$, where $p=y-\nu t$. Using these variables, the resulting submodel of~\eqref{twolayer} is:
\[
    (\nu +\sigma p)v^1_{ppp} + 2b\beta p = 0, \quad (\nu +\sigma p)v^2_{ppp} - 2F(\nu +\sigma p)v^2_p + 4F\kappa = 0,
\]
which is a decoupled system of third order linear ODEs. The case of $b=0$ is trivial and will not be considered here. For $b\ne0$ the general solution is
\begin{align*}
v^1 = {}&\frac{\nu}{\sigma}\beta\ln(\nu +\sigma p)\left(p^2+2\frac{\nu}{\sigma}p+\frac{\nu^2}{\sigma^2}\right)
-\frac13\beta\left(p^3+\frac{9\nu}{2\sigma}p^2+\frac{6\nu^2}{\sigma^2}p+\frac{3\nu^3}{2\sigma^3}\right) \\
&+ c_1p^2+c_2p+c_3, \\
v^2 = {}&\frac{\kappa}{\sigma}\Big(2\ln(\nu +\sigma p)
+\mathrm{Ei}\left(\tfrac{\sqrt{2F}}{\sigma}(\nu +\sigma p)\right)e^{\frac{\sqrt{2F}}{\sigma}(\nu +\sigma p)}
+\mathrm{Ei}\left(-\tfrac{\sqrt{2F}}{\sigma}(\nu +\sigma p)\right)e^{-\frac{\sqrt{2F}}{\sigma}(\nu +\sigma p)}\Big)\\
&+ c_4e^{\sqrt{2F}p}+c_5e^{-\sqrt{2F}p}+c_6,
\end{align*}
where $\mathrm{Ei}(z)=\int_z^\infty t^{-1}e^{-t}\ddd t$ denotes the exponential integral. In terms of the $\psi^+$/$\psi^-$ variables this solution is the superposition of some polynomial with an exponential function in $y$-direction. From the meteorological point of view, this solution does not seem to be relevant. 

\subsection{Subalgebra \texorpdfstring{$\mathcal A^2_3$}{A23}}

Using the barotropic/baroclinic variables, the ansatz for reduction under this subalgebra is $\psi^+=v^1(p)+2\mu x$, $\psi^-=v^2(p)+2\kappa t + 2\rho x$, where $p=y-\nu t$. The reduced system then is
\begin{gather*}\samepage
    (\nu-\mu)v^1_{ppp} - \rho v^2_{ppp} - 2\beta\mu = 0, \\
    (\nu-\mu)v^2_{ppp} - \rho v^1_{ppp} -2F(\nu-\mu)v^2_p - 2F\rho v^1_p +4F\kappa -2\beta\rho = 0.
\end{gather*}
This system is now completely integrable. The case $\rho=0$ and $\mu\ne\nu$ leads to a decoupled system of equations which can be integrated easily. The same also holds in the case $\rho\ne0$ and $\nu=\mu$. Hence, we focus on the case where $\rho\ne0$ and~$\nu$ and~$\mu$ are arbitrary. The first equation can be integrated at once three times, yielding the relation between $v^1$ and $v^2$, given by
\[
    v^2 = \frac{1}{\rho}\left((\nu-\mu)v^1 -\frac13\beta\mu p^3+c_1p^2+c_2p+c_3\right),
\]
where $c_1$, $c_2$ and $c_3$ are arbitrary constants. Then, integrating once the second equation and substituting the expression for $v^2$ produces a second order ordinary differential equation with constant coefficients, from which we determine~$v^1$:
\[
    v^1 = A^1\sin\sqrt{\frac{\gamma^2}{\gamma^1}}p + A^2\cos\sqrt{\frac{\gamma^2}{\gamma^1}}p - \frac{1}{\gamma^2}\left(\delta^3p^3+\delta^2p^2+\delta^1p+\delta^0\right)+ \frac{\gamma^1}{(\gamma^2)^2}(6\delta^3p+2\delta^2),
\]
where
\begin{align*}
     &\gamma^2 = -2F\left(\frac{1}{\rho}(\nu-\mu)^2+\rho\right),&&\gamma^1 = \frac{1}{\rho}(\nu-\mu)^2-\rho,  \\
     &\delta^3 = \frac{2F\beta\mu }{3\rho}(\nu-\mu),&& \delta^2 = -\frac{2c_1F}{\rho}(\nu-\mu), \\
     &\delta^1 = -2\left(\frac{\beta\mu+c_2F}{\rho}(\nu-\mu)-(2F\kappa-\beta\rho)\right),&& \delta^0 = \frac{2(c_1-c_3F)}{\rho}(\nu-\mu)+c_4.
\end{align*}
provided that $\gamma^2/\gamma^1>0$. In the particular case of $\nu=\mu$, which leads to a considerable simplification of the above solution, this condition is verified. In the case of $\gamma^2/\gamma^1<0$ we can find a solution in terms of exponential functions, which is not presented here. Plugging this solution into the ansatz for the original unknown functions, the above solution gives the combination of a traveling wave in $y$-direction with a third order time-dependent polynomial in $x$ and $y$. For usual fixed boundaries in north--south direction, this is an unphysical solution. 

\subsection{Subalgebra \texorpdfstring{$\mathcal A^2_4$}{A24}}

An appropriate ansatz for this subalgebra is 
\[\psi^+=v^1(p)-f'y^2-2g(fy-x),\quad \psi^-=v^2(p)+2\kappa y-2\rho(fy-x),\quad p=t,\] 
where we have again employed the barotropic/baroclinic variables. This ansatz enables reduction of~\eqref{twolayer2} to the system
\[
     f'' - \beta g = 0,\quad v^2_p +2\kappa g - \frac{\beta\rho}{F} = 0.
\]
The general solution of this system is
\[
    v^1 = \theta(p), \quad v^2 = -\frac{2\kappa}{\beta}f'+\frac{\beta\rho}{F}p+c,
\]
where $\theta$ is an arbitrary function of $p$ and $c$ is an arbitrary constant. The first equation of the reduced system is the compatibility condition of the initial system~\eqref{twolayer2} and the invariant surface condition corresponding to the subalgebra $\mathcal{A}^2_4$. It implies the constraint $g=f''/\beta$ for the parameter functions $f$ and $g$, which is the necessary and sufficient condition for system~\eqref{twolayer2} to have solutions invariant with respect to the subalgebra $\mathcal{A}^2_4$. This solution has no obvious physical importance in dynamic meteorology. The reason is that this is a simple polynomial solution with time-dependent coefficients. We note that the function $\theta$ can be set to zero due to gauging of the stream functions generated by $\ZZ(g)$.

\section{Invariant reduction of boundary value problems}\label{sec:boundaryQG}

In this section, we aim to discuss admitted symmetries in the presence of boundaries. It is commonly assumed that group-invariant solutions may describe the behavior of a system that is far away from boundaries and hence a consideration of restrictions imposed by boundaries is usually omitted. However, as shown e.g.\ in~\cite{hirs97Ay,land96Ay}, there may be situations where the system without boundaries is not simply the limit of a system with very distant boundaries. Consequently, consideration of boundaries may be necessary even for a conceptual understanding of the model evolution. Moreover, as noted in the two previous sections, some of the group-invariant solutions corresponding to the optimal sets of inequivalent subalgebras give rise to unphysical solutions due to a violation of boundaries. We now compute those symmetries that are admitted by the boundaries and hence discuss which solutions could be compatible with the boundary value problem.

In the atmospheric sciences, for equations on the $\beta$-plane commonly a channel flow is assumed, which implies rigid boundaries $y=0$ and $y=Y$ in north--south direction. In east--west direction, one usually assumes periodic boundaries at $x=-L$ and $x=L$ or an infinitely extended domain. In this setting, imposed conditions for the two-layer model are
\begin{equation}\label{eq:boundaryconditions}
    \pdl{\psi_i}{x} = 0, \quad \pdl{}{t}\frac{1}{2L}\int\limits_{-L}^L\pdl{\psi_i}{y}\ddd x=0 \qquad \textup{for} \qquad y \in \{0,Y\}.
\end{equation}
The second condition implies conservation of circulation at the boundaries. According to~\cite{blum89Ay}, for a boundary value problem to be invariant, three conditions must be satisfied: (i) invariance of the equation, (ii) invariance of the domain, (iii) invariance of the values on the boundaries. The first condition was already established in Section~\ref{sec:symmetriesQG}, so it remains to verify (ii) and (iii). Although it is possible to solve this problem on the stage of the Lie algebra using the infinitesimal method~\cite{blum89Ay}, we find it more comfortable to work with the finite group transformations. The most general form of a continuous symmetry transformation of~\eqref{twolayer} is given by
\[
    (t,x,y,\psi_1,\psi_2) \mapsto (t+\ve_1, x+f, y+\ve_2,\psi_1-f'y+g+\ve_3,\psi_2-f'y+g-\ve_3),
\]
where $\ve_1$, $\ve_2$ and~$\ve_3$ are arbitrary constants and $f$ and~$g$ are arbitrary smooth functions of~$t$.
This transformation has to preserve both the domain and the boundary values. We now discuss the channel flow with three possibilities for boundaries in east--west direction.

\medskip

\noindent\textbf{Infinite domain.} If there are no sidewalls in east--west direction, that is $L\to\infty$, the most general symmetry preserving the boundary value problem is given by
\[
    (t,x,y,\psi_1,\psi_2) \mapsto (t+\ve_1, x+h, y,\psi_1-h'y+g+\ve_3,\psi_2-h'y+g-\ve_3).
\]
where $h$ is an arbitrary function of $t$ with $h''=0$. Hence, we have $h=\ve_4 t+\ve_5$, where $\ve_3$ and~$\ve_5$ are arbitrary constants, leading to the important result that Galilean boosts preserve the boundary value problem.

\medskip

\noindent\textbf{Periodic boundaries.} For periodic boundaries, we have $\psi_i(t,-L,y)=\psi_i(t,L,y)$. Similar calculations as above imply that the boundary-preserving symmetry group is the same as for an infinite-domain.

This shows that the Rossby wave solution is admitted by the boundary value problem. This may serve as a ``symmetry explanation'' for the prominent occurrence of this solution in geophysical fluid dynamics.

\medskip

\noindent\textbf{Limited domain.} The model of the limited domain in east--west direction is very natural in oceanography but can be also realized in the atmospheric sciences as flow in a mountainous region. For the purpose of simplicity, we assume a rectangular domain. Besides~\eqref{eq:boundaryconditions}, this setting requires the additional conditions
\[
        \pdl{\psi_i}{y} = 0, \quad \pdl{}{t}\frac{1}{Y}\int\limits_{0}^Y\pdl{\psi_i}{x}\ddd y=0 \qquad \textup{for} \qquad x \in \{-L,L\},
\]
where $L$ and $Y$ denote the length and the width of the rectangle, respectively. Symmetries that are compatible with this boundary value problem are
\[
    (t,x,y,\psi_1,\psi_2) \mapsto (t+\ve_1, x, y,\psi_1+g+\ve_3,\psi_2+g-\ve_3).
\]
Hence, the only group-invariant solution that can be realized on this domain is a stationary solution.

\section{Conclusion}\label{sec:conclusionQG}

In this paper, we consider the baroclinic two-layer model from the viewpoint of Lie symmetries. 
The maximal Lie invariance algebra and the complete point symmetry group including both continuous and discrete symmetries are computed. 
We construct exhaustive sets of one- and two-dimensional subalgebras which are inequivalent with respect to the adjoint action and carry out the corresponding Lie reductions in one and two variables. This completely solves the classical problem of Lie reduction for the two-layer model. The procedure leads to various classes of exact solutions, some of which are well known in the atmospheric sciences, including barotropic and baroclinic Rossby waves. Finally, also the two-layer boundary value problem is investigated in the light of admitted Lie symmetries. We obtain a result, which is similar to that obtained in~\cite{pukh72Ay} (see also \cite[pp.~379]{ovsi82Ay}) for the Navier--Stokes equations, namely that periodic boundary conditions admit Galilean boosts as symmetry transformations.

Although there are still a large number of obviously unphysical solutions, the study of the classical Lie problem is a necessary first step for the consideration of partially invariant solutions and nonclassical symmetries, which we save for future investigations. In addition, these solutions are of undeniable value when it comes to a numerical implementation of the two-layer equations, which can employ several kinds of boundary conditions. The solutions constructed can be used as benchmark tests to assess the quality of the numerical scheme involved by addressing issues such as convergence rates and the reproduction of correct phase space velocities of wave-like solutions. Moreover, computing differential invariants of subalgebras of the maximal Lie invariance algebra determined in this paper, one can extend the set of exact solutions, e.g.\ by the construction of differentially invariant solutions~\cite{golo04Ay}.

From the mathematical point of view it is interesting to note the following properties of each of the above inequivalent Lie reductions, excluding that obtained using the subalgebra $\mathcal A^1_1$:  
The associated reduced system is linear and has a simpler form in the ``barotropic/baroclinic variables''. 
The coupling of the reduced equations written in terms of ``barotropic/baroclinic variables'' is only under the presence of the ``baroclinic'' operator $\mathcal F=\p_{\psi^1}-\p_{\psi^2}$ in the subalgebra. 
The reduced equation for the baroclinic part of the model is structurally more complicated than that for the barotropic part. 

The linearity of the reduced system is especially beneficial in the second and third cases of reduction in one variable.
The linear superposition principle available for the resulting systems of two-dimensional partial differential equations
allows one to generate wide sets of exact solutions by linear combining of different solutions and substituting them back into the ansatz for the original unknown functions. Furthermore, there exist a number of special techniques for finding exact solutions of linear partial differential equations. 
One of them is presented in Appendix~\ref{sec:extendedsetofsolutionsQG}.

Since the two-layer model is only capable of resolving the barotropic mode and the first baroclinic mode, it would be interesting to study symmetry properties of multi-layer models. This would be a preliminary step on the way to the investigation of a continuously stratified atmosphere. On the other hand, from the standpoint of application, a deeper investigation of layer models may even be more important than the three-dimensional system of governing equations. This is true since numerical utilization of these equations calls for some discretization, hence naturally leading back to the model of a layered atmosphere. Therefore, the present investigation of the two-layer model may not only be interesting for historical reasons.

\subsection*{Acknowledgements}

AB is a recipient of a DOC-fellowship of the Austrian Academy of Sciences. The research of ROP was supported by the Austrian Science Fund (FWF), project P20632.

\appendix

\section{Extended Lie reduction of linear PDEs}\label{sec:extendedsetofsolutionsQG}

For clarity of the presentation we confine our consideration to the case of a single differential equation with a single dependent variable. 

Consider a linear partial differential equation $\mathcal L$: $Lu=0$ in the unknown function~$u$ of $n$~independent variables $x=(x_1,\dots,x_n)$,
where $L$ is the associated linear differential operator.
In what follows we use the summation convention for repeated indices.
The indices~$i$, $j$ and~$k$ run from~1 to~$n$, the indices~$a$ and~$b$ run from 1 to~$m$.

Suppose that the equation $\mathcal L$ possesses a nontrivial Lie symmetry operator~$Q_0$ of the form
$Q_0=\xi^i(x)\p_{x_i}+\eta(x)u\p_u$, where $\xi^i\xi^i\ne0$.
Then for an arbitrary constant~$\lambda$ the equation $\mathcal L$ obviously possesses also the vector field
$Q_\lambda=Q_0+\lambda u\p_u$ as a nontrivial Lie symmetry operator.
By $\widehat Q_\lambda$ we denote the differential operator which acts on functions of~$x$ and is associated with
the operator~$Q_\lambda$, i.e., $\widehat Q_\lambda=-\xi^i(x)\p_{x_i}+\eta(x)+\lambda$.
For any $m\in\mathbb N$ the differential function $(\widehat Q_\lambda)^mu$ is well known to be a characteristic of
a generalized symmetry of~$\mathcal L$,
and hence any associated generalized ansatz reduces the equation~$\mathcal L$
to a system of $m$ linear differential equations in $m$ new unknown functions of $n-1$ new independent variables
invariant with respect to the operator~$Q_\lambda$.
In order to construct an ansatz, we should integrate the partial differential equation $(\widehat Q_\lambda)^mu=0$.
The general solution of this equation gives the ansatz
\begin{equation}\label{EqGenAnsatzForLinEqs}
u=h(x)e^{\lambda\theta}\sum_{a=1}^m \varphi^a(\omega)\frac{\theta^{m-a}}{(m-a)!},
\end{equation}
where
$\omega=(\omega^1(x),\dots,\omega^{n-1}(x))$ is a tuple of functionally independent solutions of the equation $\xi^iu_{x_i}=0$,
which are assumed to be invariant independent variables,
$\theta=\theta(x)$ is a particular solution of the equation $\xi^iu_{x_i}=1$,
$h=h(x)$ is a particular nonvanishing solution of the equation $\xi^iu_{x_i}=\eta u$ and
$\varphi^a=\varphi^a(\omega)$ play the role of new unknown functions.

In view of the Lie invariance with respect to the operator~$Q_0$,
the equation~$\mathcal L$ is mapped by the point transformation
\[
\tilde x_1=\omega^1(x),\quad \dots,\quad \tilde x_{n-1}=\omega^{n-1}(x),\quad  \tilde x_n=\theta(x),\quad \tilde u =\frac u{h(x)}
\]
to the equation $\tilde L\tilde u=0$, where the coefficients of the operator~$\tilde L$ do not depend on~$\tilde x_n$.
In the new variables $(\tilde x, \tilde u)$ the ansatz~\eqref{EqGenAnsatzForLinEqs} takes the form
\begin{equation}\label{EqGenAnsatzForLinEqsTransformed}
\tilde u=e^{\lambda \tilde x_n}\sum_{a=1}^m \varphi^a(\tilde x_1,\dots,\tilde x_{n-1})\frac{\tilde x_n^{m-a}}{(m-a)!}.
\end{equation}
After substituting the ansatz~\eqref{EqGenAnsatzForLinEqsTransformed} into the transformed equation $\tilde L\tilde u=0$,
dividing the resulting equation by $e^{\lambda \tilde x_n}$ and subsequently splitting with respect to different powers of the variable~$\tilde x_n$,
we obtain, at least for the general value of~$\lambda$, the system~$\mathcal R$ of $m$~differential equations with respect to the functions~$\varphi^a$
in $n-1$ independent variables $(\tilde x_1,\dots,\tilde x_{n-1})$.
In singular cases, for certain values of~$\lambda$ some of the equations are identities.
The same reduction is obtained by the substitution of the ansatz~\eqref{EqGenAnsatzForLinEqs} into the initial equation~$\mathcal L$.

If the basic field is real, we can consider complex values of~$\lambda$, construct the corresponding complex exact solution
and then take its real and imagine parts in order to obtain real solutions.

The above consideration has a nice interpretation within the framework of Lie symmetries.
Introducing the new dependent variables $v^a=(\widehat Q_\lambda)^{m-a}u$,
instead of the single $m$th order linear partial differential equation $(\widehat Q_\lambda)^mu=0$ for finding a generalized ansatz,
we obtain the system of $m$ first order linear partial differential equations
\[
\widehat Q_\lambda v^1=0, \quad \widehat Q_\lambda v^a=v^{a-1}, \quad a=2,\dots,m.
\]
As $Q_\lambda$ is a Lie symmetry operator of the equation~$\mathcal L$, each function~$v^a$ satisfies this equation.
To give the interpretation, we consider the system~$\mathcal S$ of $m$ copies of the initial equation~$\mathcal L$
\[
Lv^1=0,\quad \dots,\quad Lv^m=0.
\]
This system obviously possesses the operators $\bar Q_0=\xi^i(x)\p_{x_i}+\eta(x)v^a\p_{v^a}$ and $v^b\p_{v^a}$ as its Lie symmetry operators.
Consider a linear combination of these operators, $\bar Q_\Lambda=\bar Q_0+\Lambda_{ab}v^b\p_{v^a}$,
which is also a Lie symmetry operator of the system~$\mathcal S$.
Here and in what follows $\Lambda_{ab}$ are constants.
Up to the equivalence generated by adjoint action of the Lie symmetry group of~$\mathcal S$ on the corresponding Lie invariance algebra
and due to the linear superposition principle,
we can assume without loss of generality that the matrix $\Lambda=(\Lambda_{ab})$ is the single $m\times m$ Jordan block with an eigenvalue~$\lambda$,
\[
\Lambda=J_\lambda^m=\arraycolsep=0ex
\left(\begin {array}{cccccc}
\lambda & 1 & 0 & 0 & \cdots & 0 \\
0 & \lambda & 1 & 0 & \cdots & 0 \\
0 & 0 & \lambda & 1 & \cdots & 0 \\
\cdots & \cdots & \cdots & \cdots & \cdots & \cdots\\
0 & 0 & 0 & 0 & \cdots & 1 \\
0 & 0 & 0 & 0 & \cdots & \lambda
\end {array}\right).
\]
The invariant surface condition for the operator~$\bar Q_\Lambda$ with $\Lambda=J_\lambda^m$ consists of the equations
\[
\xi^iv^1=(\eta+\lambda)v^1,\quad \xi^iv^a=(\eta+\lambda)v^a+v^{a-1},\quad a=2,\dots,m,
\]
and an ansatz constructed with this operator has the form
\begin{equation}\label{EqGenAnsatzForLinEqsIterated}
v^a=h(x)e^{\lambda\theta}\sum_{b=1}^a \varphi^b(\omega)\frac{\theta^{a-b}}{(a-b)!},
\end{equation}
where the notation from the ansatz~\eqref{EqGenAnsatzForLinEqsTransformed} is used.
According to the general theory of Lie reduction~\cite{olve86Ay},
the ansatz~\eqref{EqGenAnsatzForLinEqsIterated} necessarily reduces the system~$\mathcal S$ to a system in the functions~$\varphi^a$,
which obviously coincides with the system~$\mathcal R$ obtained by reducing the single equation~$\mathcal L$
using the generalized ansatz~\eqref{EqGenAnsatzForLinEqs}.

If the equation~$\mathcal L$ is considered over the real field and the eigenvalue~$\lambda$ is complex,
it is not necessary to pass to the associated real Jordan block.
As above, we can find the corresponding complex exact solution
and then take its real and imagine parts in order to construct real solutions.
This additionally justifies the usage of complex values of~$\lambda$ in the real case.

\begin{example*}
For the general value of~$A$, the second equation of system~\eqref{eq:redsublinear} admits
only one independent nontrivial Lie symmetry operator, $\p_p$.
Consider a system of $m$ copies of this equation:
\begin{gather}\label{eq:systemlinear}
v^a_{ppq}-2(Av^a)_q + v^a_p = 0,
\end{gather}
where for simplicity we have omitted bars over the variables and scaled $F=1$ and $\beta=1$. This system admits the Lie symmetry operator
$\bar Q_\Lambda =\p_p +\Lambda_{ab}v^b\p_{v^a}$, where $\Lambda=J_\lambda^m$.
The invariant surface condition associated with $\bar Q_\Lambda$ reads
\[v^1_p=\lambda v^1,\quad v^a_p=\lambda v^a+v^{a-1},\quad a=2,\dots,m.\]
Its general solution provides us with an appropriate ansatz for Lie reduction:
\[
v^a = \exp(\lambda p)\sum_{b=1}^a\varphi^b(\omega)\frac{p^{a-b}}{(a-b)!},
\]
where $\omega=q=t$ is the invariant independent variable.
Substituting this ansatz into system~\eqref{eq:systemlinear} yields the system of ordinary differential equations for~$\varphi^a$
\begin{align*}
    &\mathrm L\varphi^1 = 0, \\
    &\mathrm L\varphi^2 + 2\lambda \varphi^1_q + \varphi^1 = 0, \\
    &\mathrm L\varphi^k + 2\lambda \varphi^{k-1}_q + \varphi^{k-1} + \varphi^{k-2}_q = 0, \qquad k=3,\dots,m,
\end{align*}
where the operator $\mathrm L$ is given by $\mathrm L:=(\lambda^2-2A)\p_q - 2A_q + \lambda$.

The solution of the above system is:
\begin{align*}
    &\varphi^1 = c_1e^{-\zeta},  \\
    &\varphi^2 = c_2e^{-\zeta}  + e^{-\zeta} \int \frac{\varphi^1+2\lambda\varphi^1_q}{2A-\lambda^2}e^\zeta\,\ddd q, \\
    &\varphi^k = c_ke^{-\zeta}  + e^{-\zeta} \int \frac{\varphi^{k-1}+2\lambda\varphi^{k-1}_q+\varphi^{k-2}_q}{2A-\lambda^2}e^\zeta\,\ddd q, \qquad k=3,\dots,m,
\end{align*}
where
\[
    \zeta =\int\frac{2A_q-\lambda}{2A-\lambda^2}\ddd q.
\]

In the special cases $A=\const$ and $A=\varkappa(q+C_0)^2$, where $\varkappa,C_0=\const$,
we can make more generalized reductions of the equation under consideration, which involves operators extending the Lie invariance algebra
of the general case, cf.\ Section~\ref{sec:SubalgebraA12}.
\end{example*}


{\footnotesize
\begin{thebibliography}{10}\itemsep=1ex

\bibitem{abra06Ay}
B.~Abraham-Shrauner and K.~S. Govinder.
\newblock Provenance of type {II} hidden symmetries from nonlinear partial differential equations.
\newblock {\em J. Nonlin. Math. Phys.}, 13(4):612--622, 2006.

\bibitem{andr98Ay}
V.~K. Andreev, O.~V. Kaptsov, V.~V. Pukhnachov, and A.~A. Rodionov.
\newblock {\em Applications of group-theoretical methods in hydrodynamics}.
\newblock Kluwer, Dordrecht, 1998.

\bibitem{bihl09By}
A.~Bihlo and R.~O. Popovych.
\newblock {Symmetry justification of Lorenz' maximum simplification}.
\newblock {\em Nonlin. Dyn.}, 61(1--2):101--107, 2009.
\newblock arXiv:0805.4061v2.

\bibitem{bihl09Ay}
A.~Bihlo and R.~O. Popovych.
\newblock Lie symmetries and exact solutions of the barotropic vorticity equation.
\newblock {\em J. Math. Phys.}, 50:123102 (12 pages), 2009.
\newblock arXiv:0902.4099.

\bibitem{bihl09Cy}
A.~Bihlo and R.~O. Popovych.
\newblock Symmetry analysis of barotropic potential vorticity equation.
\newblock {\em Comm. Theor. Phys.}, 52(4):697--700, 2009.
\newblock arXiv:0811.3008v2.

\bibitem{bihl10Cy}
A.~Bihlo and R.~O. Popovych.
\newblock Point symmetry group of the barotropic vorticity equation,  
\newblock 2010, arXiv:1009.1523, 13 pp.

\bibitem{blum89Ay}
G.W. Bluman and S.~Kumei.
\newblock {\em {Symmetries and differential equations}}.
\newblock Springer, New York, 1989.

\bibitem{carm00Ay}
J.~Carminati and K.~Vu.
\newblock Symbolic computation and differential equations: {L}ie symmetries.
\newblock {\em J. Symb. Comput.}, 29(1):95--116, 2000.

\bibitem{fush94Ay}
W.~I. Fushchych and R.~O. Popovych.
\newblock Symmetry reduction and exact solutions of the {N}avier--{S}tokes equations.
\newblock {\em J. Nonl. Math. Phys.}, 1(1--2):75--113,156--188, 1994.
\newblock arXiv:math-ph/0207016.

\bibitem{golo04Ay}
S.~V. Golovin.
\newblock {Applications of the differential invariants of infinite dimensional groups in hydrodynamics}.
\newblock {\em Commun. Nonlinear Sci. Numer. Simul.}, 9(1):35--51, 2004.

\bibitem{head93Ay}
A.~K. Head.
\newblock {LIE}, a {PC} program for {L}ie analysis of differential equations.
\newblock {\em Comput. Phys. Comm.}, 77(2):241--248, 1993.
\newblock (See also http://www.cmst.csiro.au/LIE/LIE.htm).

\bibitem{hirs97Ay}
P.~Hirschberg and E.~Knobloch.
\newblock {Mode interactions in large aspect ratio convection}.
\newblock {\em J. Nonlinear Sci.}, 7(6):537--556, 1997.

\bibitem{holt04Ay}
J.~R. Holton.
\newblock {\em An introduction to dynamic meteorology}.
\newblock Acad. Press, Amsterdam, 2004.

\bibitem{hydo00By}
P.~E. Hydon.
\newblock How to construct the discrete symmetries of partial differential equations.
\newblock {\em Eur. J. Appl. Math.}, 11(5):515--527, 2000.

\bibitem{ibra95Ay}
N.~H. Ibragimov, A.~V. Aksenov, V.~A. Baikov, V.~A. Chugunov, R.~K. Gazizov, and A.~G. Meshkov.
\newblock {\em CRC handbook of Lie group analysis of differential equations.
  Vol. 2. Applications in engineering and physical sciences. Edited by N. H. Ibragimov}.
\newblock CRC Press, Boca Raton, 1995.

\bibitem{katk65Ay}
V.~L. Katkov.
\newblock A class of exact solutions of the equation for the forecast of the geopotential.
\newblock {\em Izv. Acad. Sci. USSR Atmospher. Ocean. Phys.}, 1:630--631, 1965.

\bibitem{katk66Ay}
V.~L. Katkov.
\newblock Exact solutions of the geopotential forecast equation.
\newblock {\em Akad. Nauk SSSR Ser. Fiz. Atmosfer. i Okeana}, 2:1193, 1966.

\bibitem{land96Ay}
A.~S. Landsberg and E.~Knobloch.
\newblock {Oscillatory bifurcation with broken translation symmetry}.
\newblock {\em Phys. Rev. E}, 53(4):3579--3600, 1996.

\bibitem{mele04Ay}
S.~V. Meleshko.
\newblock A particular class of partially invariant solutions of the {N}avier--{S}tokes equations.
\newblock {\em Nonlin. Dyn.}, 36(1):47--68, 2004.

\bibitem{mele05Ay}
S.~V. Meleshko.
\newblock {\em Methods for constructing exact solutions of partial differential equations}.
\newblock Mathematical and analytical techniques with applications to
  engineering. Springer, New York, 2005.

\bibitem{olve86Ay}
P.~J. Olver.
\newblock {\em Application of Lie groups to differential equations}.
\newblock Springer, New York, 2000.

\bibitem{ovsi82Ay}
L.~V. Ovsiannikov.
\newblock {\em Group analysis of differential equations}.
\newblock Acad. Press, New York, 1982.

\bibitem{pedl87Ay}
J.~Pedlosky.
\newblock {\em Geophysical fluid dynamics}.
\newblock Springer, New York, 1987.

\bibitem{phil54Ay}
N.~A. Phillips.
\newblock {Energy transformations and meridional circulations associated with simple baroclinic waves in a two-level, quasi-geostrophic model}.
\newblock {\em Tellus}, 6(3):273--286, 1954.

\bibitem{poly02Ay}
A.~D. Polyanin.
\newblock {\em Handbook of linear partial differential equations for engineers and scientists}.
\newblock Chapman \& Hall/CRC, Boca Raton, 2002.

\bibitem{popo00Ay}
H.~V. Popovych.
\newblock {L}ie, partially invariant and nonclassical submodels of the {E}uler equations.
\newblock In {\em Proceedings of Institute of Mathematics of NAS of Ukraine},
  volume 43/1, pages 178--183, Kyiv, 2002.

\bibitem{popo10Cy}
R.~O. Popovych and A.~Bihlo.
\newblock {Symmetry preserving parameterization schemes}.
\newblock 2010, arXiv:1010.3010, 36~pp.

\bibitem{popo05Ay}
R.~O. Popovych, V.~M. Boyko, M.~O. Nesterenko, and M.~W. Lutfullin.
\newblock {Realizations of real low-dimensional Lie algebras.}
\newblock {\em J. Phys. A}, 36(26):7337--7360, 2003. 
\newblock (See also the extended and revised version arXiv:math-ph/0301029v7, 2005, 39 pp.)

\bibitem{popo04Ay}
R.~O. Popovych and N.~M. Ivanova.
\newblock New results on group classification of nonlinear diffusion--convection equations.
\newblock {\em J. Phys. A}, 37(30):7547--7565, 2004.

\bibitem{popo10Ay}
R.~O. Popovych, M.~Kunzinger, and H.~Eshraghi.
\newblock {Admissible transformations and normalized classes of nonlinear Schr\"{o}dinger equations.}
\newblock {\em Acta Appl. Math.}, 109(2):315--359, 2010.

\bibitem{prok05Ay}
Prokhorova M.,
The structure of the category of parabolic equations, 2005, arXiv:math.AP/0512094, 24~pp.


\bibitem{pukh72Ay}
V.~V.~Pukhnachov.
\newblock {Invariant solutions of Navier--Stokes equations describing free boundary motion}.
\newblock {\em Dokl. Akad. Nauk SSSR}, 202(2):302--305, 1972.

\end{thebibliography}

}
\end{document}